\documentclass{article}
\usepackage{amsmath,amsthm,amssymb}

\allowdisplaybreaks[3]

\numberwithin{equation}{section}

\newtheorem{theorem}{Theorem}[section]
\newtheorem{corollary}[theorem]{Corollary}
\newtheorem{lemma}[theorem]{Lemma}
\newtheorem{proposition}[theorem]{Proposition}
\theoremstyle{remark}
\newtheorem{remark}[theorem]{Remark}
\newtheorem{example}[theorem]{Example}

\newcommand\R{{\mathbb R}}
\newcommand\X{{\R^d}}
\newcommand\N{{\mathbb N}}

\newcommand\F{{\mathcal F}}
\newcommand\B{{\mathcal B}}
\newcommand\Bc{\B_{\mathrm{b}}}
\newcommand\Bbs{B_{\mathrm{bs}}}
\newcommand\Fc{{\F_{\mathrm{cyl}}}}
\renewcommand\L{{\mathcal L}}
\newcommand\K{{\mathcal K}}

\renewcommand\a{{\alpha}}
\newcommand\aC{{\a C}}
\newcommand\La{\Lambda}
\newcommand\la{\lambda}
\newcommand\Ga{\Gamma}
\newcommand\ga{\gamma}
\newcommand\eps{\varepsilon}
\newcommand{\1}{1\!\!1}
\newcommand\n{{|\eta|}}
\newcommand\ren{{\eps, \, \mathrm{ren}}}
\newcommand\lv{\left\vert}
\newcommand\rv{\right\vert}
\newcommand\lV{\left\Vert}
\newcommand\rV{\right\Vert}
\newcommand\lu{\left\langle}
\newcommand\ru{\right\rangle}
\newcommand\llu{\lu\!\lu}
\newcommand\rru{\ru\!\ru}
\newcommand\lluu{\lu\!\!\lu}
\newcommand\rruu{\ru\!\!\ru}
\newcommand\KK{\overline{\K_\aC}}
\newcommand\hT{{\hat{T}}}
\newcommand\hQ{{\hat{Q}}}
\newcommand\hP{{\hat{P}}}
\newcommand\hL{{\hat{L}}}

\newcommand\goto{\rightarrow}
\newcommand\hLrenadj{\hL_\ren^\ast}
\newcommand\adot{{\odot\a}}
\newcommand\esssup{\mathop{\mathrm{ess\,sup}}}
\newcommand\e{{(\eps)}}

\begin{document}

\title{Vlasov scaling for the Glauber dynamics in~continuum}

\author{Dmitri Finkelshtein\thanks{Institute of Mathematics,
National Academy of Sciences of Ukraine, Kyiv, Ukraine ({\tt
fdl@imath.kiev.ua}).} \and Yuri Kondratiev\thanks{Fakult\"{a}t
f\"{u}r Mathematik, Universit\"{a}t Bielefeld, 33615 Bielefeld,
Germany ({\tt kondrat@math.uni-bielefeld.de})} \and Oleksandr
Kutoviy\thanks{Fakult\"{a}t f\"{u}r Mathematik, Universit\"{a}t
Bielefeld, 33615 Bielefeld, Germany ({\tt
kutoviy@math.uni-bielefeld.de}).}}

\date{}

\maketitle

\begin{abstract}
We consider Vlasov-type scaling for the Glauber dynamics in
continuum with a positive integrable potential, and construct
rescaled and limiting evolutions of correlation functions.
Convergence to the limiting evolution for the positive density
system in infinite volume is shown. Chaos preservation property of
this evolution gives a possibility to derive a non-linear
Vlasov-type equation for the particle density of the limiting
system.
\end{abstract}

\section{Introduction}

Kinetic equations are a useful approximation for the description of
dynamical processes in multi-body systems, see, e.g., the reviews by
H.Spohn \cite{Spo1980}, \cite{Spo1991}. Among them, the Vlasov
equation has important role in physics (in particular, physics of
plasma). It describes the Hamiltonian motion of an infinite particle
system in the mean field scaling limit when the influence of weak
long-range forces is taken into account. The convergence of the
Vlasov scaling limit was shown rigorously by W.Braun and K.Hepp
\cite{BH1977} (for the Hamiltonian dynamics) and by R.L.Dobrushin
\cite{Dob1979} (for more general deterministic dynamical systems).
However, the resulting Vlasov-type equations for particle densities
are considered in classes of integrable functions (or, in the weak
form, of finite measures). This, in fact, restricts us to the case
of finite volume systems or systems with zero mean density in an
infinite volume. Detailed analysis of Vlasov-type equations for
integrable functions is presented in the recent paper by V.V.Kozlov
\cite{Koz2008}.

In \cite{FKK2010a}, we proposed a general approach to study the
Vlasov-type scaling for some classes of stochastic evolutions in the
continuum, in particular, for spatial birth-and-death Markov
processes. The approaches mentioned above are not applicable to
these dynamics (even in a finite volume) due to essential reasons
(see \cite{FKK2010a} for details). One of them is a possible
variation of the particle number during the evolution. More
essentially is that for these processes the possibility of their
descriptions in terms of proper stochastic evolutional equations for
particle motion is, generally speaking, absent. There are only few
works concerning general spatial birth-and-death evolutions, see
\cite{Pre1975}, \cite{HS1978}, \cite{GK2006}, \cite{GK2008},
\cite{Pen2008}, \cite{Qi2008}. However, the conditions for the
existence (in different senses) of the evolutions considered therein
are quite far from the general form.

Therefore, we looked for an alternative approach to the derivation
of kinetic Vlasov-type equations from stochastic dynamics. The
correct Vlasov limit can be easily guessed from the BBGKY hierarchy
for the Hamiltonian system, see, e.g., \cite{Spo1980}. Such a
heuristic derivation does not assume the integrability condition for
the density, but until now, it could not be made rigorously due to
the lack of detailed information about the properties of solutions
to the BBGKY hierarchy. Our approach is based on this observation
applied in a new dynamical framework. Note that we already know that
many stochastic evolutions in continuum admit effective descriptions
in terms of hierarchical equations for correlation functions which
generalize the BBGKY hierarchy from Hamiltonian to Markov setting,
see, e.g., \cite{FKO2009} and the references therein. Even more,
these hierarchical equations are often the only available technical
tools for a construction of considered dynamics \cite{KKM2008},
\cite{KKZ2006}, \cite{FKK2009}.

Developing this point of view, our scheme for the Vlasov scaling of
stochastic dynamics is based on the proper scaling of the
hierarchical equations. This scheme has also a clear interpretation
in the terms of scaled Markov generators. An application of the
considered scaling leads to the limiting hierarchy which posses a
chaos preservation property. Namely, if we start from a Poissonian
(non-homogeneous) initial state of the system, then during the time
evolution this property will be preserved. Moreover, a special
structure of the interaction in the resulting virtual Vlasov system
gives a non-linear evolutional equation for the density of the
evolving Poisson state.

The control of the convergence of Vlasov scalings for the considered
hierarchies is a quite difficult technical problem which should be
analyzed for any particular model separately.
 In the present paper,
we solve this problem for the Glauber dynamics in continuum. These
dynamics have given reversible states which are grand canonical
Gibbs measures. The corresponding equilibrium dynamics which
preserve the initial Gibbs state in the time evolution were
considered in, e.g., \cite{KL2005}, \cite{KLR2007}, \cite{KMZ2004},
\cite{FKL2007}. Note that, in applications, the time evolution of
initial state is the subject of the primary interest. Therefore, we
understand the considered stochastic (non-equilibrium) dynamics as
the evolution of initial distributions for the system. Actually, the
corresponding Markov process (provided it exists) itself gives a
general technical equipment to study this problem. Moreover, using
the techniques developed in \cite{GK2006}, it is possible to
construct this Markov process as a solution of a stochastic
differential equation. Unfortunately, this approach does not give
any information about the properties of the corresponding
correlation functions which we need for the study of Vlasov scaling
as was mentioned above.

However, we note that the transition from the micro-state evolution corresponding to
the given initial configuration  to the macro-state dynamics is the
well developed concept in the theory of infinite particle systems.
This point of view appeared initially in the framework of the
Hamiltonian dynamics of classical gases, see, e.g., \cite{DSS1989}.
Again, the lack of the general Markov processes techniques for the
considered systems makes it necessary to develop alternative
approaches to study the state evolutions in the Glauber dynamics.
Such approaches we realized in \cite{KKM2008}, \cite{KKZ2006},
\cite{FKKZ2010}, \cite{FKK2010}. The description of the time
evolutions for measures on configuration spaces in terms of an
infinite system of evolutional equations for the corresponding
correlation functions was used there. The latter system is a Glauber
evolution's analog of the famous BBGKY-hierarchy for the Hamiltonian
dynamics.

Here we extend the approximation approach proposed in
\cite{FKKZ2010}, \cite{FKK2010} to the Vlasov scaling for the
Glauber dynamics in continuum. We construct and study semigroups
corresponding to properly rescaled Markov generator of the Glauber
dynamics (Propositions~\ref{descsemigroupexist} and \ref{sun-inv}).
We prove for the integrable and bounded potential the convergence of
these semigroups to the limiting semigroup which describe Vlasov
evolution (Theorem~\ref{maintheorem}). We derive the corresponding
Vlasov-type equation from this evolution
(Theorem~\ref{Vlasovscheme}). Note that the stationary solution of
this equation will satisfied the well-known Kirkwood--Monroe
equation in the freezing theory (Remark~\ref{RemarkKirkwoodMonroe}).

\section{Glauber dynamics in continuum}

\subsection{Basic facts and notation}

Let ${\mathcal{B}}({{\mathbb{R}}^d})$ be the family of all Borel sets in ${{%
\mathbb{R}}^d}$, $d\geq 1$; ${\mathcal{B}}_{\mathrm{b}}
({{\mathbb{R}}^d})$ denotes the system of all bounded sets in
${\mathcal{B}}({{\mathbb{R}}^d})$.

The configuration space over space ${{\mathbb{R}}^d}$ consists of
all locally finite subsets (configurations) of ${{\mathbb{R}}^d}$,
namely,
\begin{equation}  \label{confspace}
\Gamma =\Gamma_{{\mathbb{R}}^d} :=\Bigl\{ \gamma \subset
{{\mathbb{R}}^d} \Bigm| |\gamma _\Lambda |<\infty, \ \mathrm{for \ all } \ \Lambda \in {\mathcal{%
B}}_{\mathrm{b}} ({{\mathbb{R}}^d})\Bigr\}.
\end{equation}
Here $\gamma_\Lambda:=\gamma%
\cap\Lambda$, and $|\cdot|$ means the cardinality of a finite set.
The space $\Gamma$ is equipped with the vague topology, i.e., the
minimal topology for which all mappings $\Gamma\ni\gamma\mapsto
\sum_{x\in\gamma}
f(x)\in{\mathbb{R}}$ are continuous for any continuous function $f$ on ${{%
\mathbb{R}}^d}$ with compact support; note that the summation in
$\sum_{x\in\gamma} f(x)$ is taken over
finitely many points of $\gamma$ which belong to the support of $f$. In \cite%
{KK2006}, it was shown that $\Gamma$ with the vague topology may be
metrizable and it becomes a Polish space (i.e., complete separable
metric space). Corresponding to this topology, the Borel $\sigma
$-algebra ${\mathcal{B}}(\Gamma )$ is the smallest $\sigma $-algebra
for which all mappings $\Gamma \ni \gamma \mapsto |\gamma_ \Lambda |\in{%
\mathbb{N}}_0:={\mathbb{N}}\cup\{0\}$ are measurable for any $\Lambda\in{%
\mathcal{B}}_{\mathrm{b}}({{\mathbb{R}}^d})$.

The space of $n$-point configurations in an arbitrary
$Y\in{\mathcal{B}}({{\mathbb{R}}^d})$ is defined by
\[
\Gamma^{(n)}_Y:=\Bigl\{  \eta \subset Y \Bigm| |\eta |=n\Bigr\} ,\quad n\in {%
\mathbb{N}}.
\]
We set also $\Gamma^{(0)}_Y:=\{\emptyset\}$. As a set,
$\Gamma^{(n)}_Y$ may be identified with the symmetrization of
\[
\widetilde{Y^n} = \Bigl\{ (x_1,\ldots ,x_n)\in Y^n \Bigm| x_k\neq
x_l \ \mathrm{if} \ k\neq l\Bigr\} .
\]

Hence one can introduce the corresponding Borel $\sigma $-algebra,
which we denote by ${\mathcal{B}}(\Gamma^{(n)}_Y)$. The space of
finite configurations in an arbitrary
$Y\in{\mathcal{B}}({{\mathbb{R}}^d})$ is defined by
\[
\Gamma_{0,Y}:=\bigsqcup_{n\in {\mathbb{N}}_0}\Gamma^{(n)}_Y.
\]
This space is equipped with the topology of disjoint unions.
Therefore, one can introduce the corresponding Borel $\sigma
$-algebra ${\mathcal{B}} (\Gamma _{0,Y})$. In the case of
$Y={{\mathbb{R}}^d}$ we will omit the index $Y$ in
the notation, namely, $\Gamma_0:=\Gamma_{0,{{\mathbb{R}}^d}}$, $%
\Gamma^{(n)}:=\Gamma^{(n)}_{{{\mathbb{R}}^d}}$.

The restriction of the Lebesgue product measure $(dx)^n$ to $\bigl(%
\Gamma^{(n)}, {\mathcal{B}}(\Gamma^{(n)})\bigr)$ we denote by
$m^{(n)}$. We
set $m^{(0)}:=\delta_{\{\emptyset\}}$. The Lebesgue--Poisson measure $%
\lambda $ on $\Gamma_0$ is defined by
\begin{equation}  \label{LP-meas-def}
\lambda :=\sum_{n=0}^\infty \frac {1}{n!}m^{(n)}.
\end{equation}
For any $\Lambda\in{\mathcal{B}}_{\mathrm{b}}({{\mathbb{R}}^d})$ the
restriction of $\lambda$ to $\Gamma _\Lambda:=\Gamma_{0,\Lambda}$
will be also denoted by $\lambda $. The space $\bigl( \Gamma,
{\mathcal{B}}(\Gamma)\bigr)$ is the projective limit of the family
of spaces $\bigl\{( \Gamma_\Lambda, {\mathcal{B}}(\Gamma_\Lambda))\bigr\}%
_{\Lambda \in {\mathcal{B}}_{\mathrm{b}} ({{\mathbb{R}}^d})}$. The
Poisson measure $\pi$ on $\bigl(\Gamma ,{\mathcal{B}}(\Gamma
)\bigr)$ is given as the projective limit of the family of measures
$\{\pi^\Lambda \}_{\Lambda
\in {\mathcal{B}}_{\mathrm{b}} ({{\mathbb{R}}^d})}$, where $%
\pi^\Lambda:=e^{-m(\Lambda)}\lambda $ is the probability measure on $\bigl( %
\Gamma_\Lambda, {\mathcal{B}}(\Gamma_\Lambda)\bigr)$. Here
$m(\Lambda)$ is
the Lebesgue measure of $\Lambda\in {\mathcal{B}}_{\mathrm{b}} ({{\mathbb{R}}%
^d})$.

For any measurable function $f:{{\mathbb{R}}^d}\rightarrow
{\mathbb{R}}$ we define a \emph{Lebesgue--Poisson exponent}
\begin{equation}  \label{LP-exp-def}
e_\lambda(f,\eta):=\prod_{x\in\eta} f(x),\quad \eta\in\Gamma_0;
\qquad e_\lambda(f,\emptyset):=1.
\end{equation}
Then, by \eqref{LP-meas-def}, for $f\in L^1({{\mathbb{R}}^d},dx)$ we
obtain $e_\lambda(f)\in L^1(\Gamma_0,d\lambda)$ and
\begin{equation}  \label{LP-exp-mean}
\int_{\Gamma_0}
e_\lambda(f,\eta)d\lambda(\eta)=\exp\Biggl\{\int_{{\mathbb{R}}^d}
f(x)dx\Biggr\}.
\end{equation}

A set $M\in {\mathcal{B}} (\Gamma_0)$ is called bounded if there exists $%
\Lambda \in {\mathcal{B}}_{\mathrm{b}} ({{\mathbb{R}}^d})$ and $N\in {%
\mathbb{N}}$ such that $M\subset \bigsqcup_{n=0}^N\Gamma
_\Lambda^{(n)}$.
The set of bounded measurable functions with bounded support we denote by $%
B_{\mathrm{bs}}(\Gamma_0)$, i.e., $G\in B_{\mathrm{bs}}(\Gamma_0)$ if $%
G\upharpoonright_{\Gamma_0\setminus M}=0$ for some bounded $M\in {\mathcal{B}%
}(\Gamma_0)$. Any ${\mathcal{B}}(\Gamma_0)$-measurable function $G$ on $%
\Gamma_0$, in fact, is a sequence of functions $\bigl\{G^{(n)}\bigr\}_{n\in{%
\mathbb{N}}_0}$ where $G^{(n)}$ is a ${\mathcal{B}}(\Gamma^{(n)})$%
-measurable function on $\Gamma^{(n)}$. We consider also the set ${{\mathcal{%
F}}_{\mathrm{cyl}}}(\Gamma )$ of \textit{cylinder functions} on
$\Gamma$. Each $F\in {{\mathcal{F}}_{\mathrm{cyl}}}(\Gamma )$ is
characterized by the following relation: $F(\gamma
)=F\upharpoonright_{\Gamma_\Lambda
}(\gamma_\Lambda )$ for some $\Lambda\in {\mathcal{B}}_{\mathrm{b}}({{%
\mathbb{R}}^d})$.

There is the following mapping from $B_{\mathrm{bs}} (\Gamma_0)$ into ${{%
\mathcal{F}}_{\mathrm{cyl}}}(\Gamma )$, which plays the key role in
our further considerations:
\begin{equation}
KG(\gamma ):=\sum_{\eta \Subset \gamma }G(\eta ), \quad \gamma \in
\Gamma, \label{KT3.15}
\end{equation}
where $G\in B_{\mathrm{bs}}(\Gamma_0)$, see, e.g.,
\cite{KK2002,Len1975,Len1975a}. The summation in \eqref{KT3.15} is
taken over all finite subconfigurations $\eta\in\Ga_0$ of the
(infinite) configuration $\gamma\in\Ga$; we denote this by the
symbol, $\eta\Subset\gamma $. The mapping $K$ is linear, positivity
preserving, and invertible, with
\begin{equation}
K^{-1}F(\eta ):=\sum_{\xi \subset \eta }(-1)^{|\eta \setminus \xi
|}F(\xi ),\quad \eta \in \Gamma_0.  \label{k-1trans}
\end{equation}
We denote the restriction of $K$ onto functions on $\Gamma_0$ by
$K_0$.

A measure $\mu \in {\mathcal{M}}_{\mathrm{fm} }^1(\Gamma )$ is
called locally absolutely continuous with respect to (w.r.t. for
short) the Poisson measure $\pi$ if for any $\Lambda \in
{\mathcal{B}}_{\mathrm{b}} ({{\mathbb{R}}^d})$ the projection of
$\mu$ onto $\Gamma_\Lambda$ is absolutely continuous w.r.t. the projection of $%
\pi$ onto $\Gamma_\Lambda$. By \cite{KK2002}, in this case, there
exists a \emph{correlation functional} $k_{\mu}:\Gamma_0 \rightarrow
{\mathbb{R}}_+$ such that for any $G\in B_{\mathrm{bs}} (\Gamma_0)$
the following equality holds
\begin{equation}  \label{eqmeans}
\int_\Gamma (KG)(\gamma) d\mu(\gamma)=\int_{\Gamma_0}G(\eta)
k_\mu(\eta)d\lambda(\eta).
\end{equation}
The restrictions $k_\mu^{(n)}$ of this functional on $\Gamma_0^{(n)}$, $n\in{%
\mathbb{N}}_0$ are called \emph{correlation functions} of the
measure $\mu$. Note that $k_\mu^{(0)}=1$.

We recall now without a proof the partial case of the well-known
technical lemma (cf., \cite{KMZ2004}) which plays very important role
in our calculations.

\begin{lemma}
\label{Minlos} For any measurable function $H:\Gamma_0\times\Gamma_0\times%
\Gamma_0\rightarrow{\mathbb{R}}$
\begin{equation}  \label{minlosid}
\int_{\Gamma _{0}}\sum_{\xi \subset \eta }H\left( \xi ,\eta
\setminus \xi ,\eta \right) d\lambda \left( \eta \right)
=\int_{\Gamma _{0}}\int_{\Gamma _{0}}H\left( \xi ,\eta ,\eta \cup
\xi \right) d\lambda \left( \xi \right) d\lambda \left( \eta \right)
\end{equation}
if only both sides of the equality make sense.
\end{lemma}

\subsection{Non-equilibrium Glauber dynamics in continuum}\label{dualconstraction}

Let $\phi:\X\goto\R_+:=[0;+\infty)$ be an even non-negative function
which satisfies the following integrability condition
\begin{equation}\label{weak_integrability}
C_\phi := \int_\X \bigl(1-e^{-\phi(x)}\bigr) dx < +\infty.
\end{equation}
For any $\ga\in\Ga$, $x\in\X\setminus\ga$ we set
\begin{equation}\label{relativeenergy}
E^\phi(x,\ga) :=\sum_{y\in\ga} \phi(x-y) \in [0;\infty].
\end{equation}

Let us define the (pre-)generator of the Glauber dynamics: for any
$F\in\Fc(\Ga )$ we set
\begin{align}
(LF)(\ga):=&\sum_{x\in\ga} \bigl[F(\ga\setminus x) -F(\ga)\bigr]
\label{genGa}
\\&+ z \int_{\X} \bigl[F(\ga\cup x)
-F(\ga)\bigr]\exp\bigl\{-E^\phi(x,\ga)\bigr\} dx, \qquad
\ga\in\Ga.\nonumber
\end{align}
Here $z>0$ is the {\it activity} parameter. Note that for any
$F\in\Fc(\Ga)$ there exists a $\La\in\Bc(\X)$ such that
$F(\ga\setminus x)=F(\ga)$ for all $x\in\ga_{\La^c}$ and $F(\ga\cup
x)=F(\ga)$ for all $x\in\La^c$; note also that
$\exp\bigl\{-E^\phi(x,\ga)\bigr\}\leq 1$, therefore, sum and
integral in \eqref{genGa} are finite.

For any fixed $C>1$ we consider the following Banach space of
${\mathcal{B}} (\Gamma_0)$-measurable functions
\[
\L _C:=\biggl\{ G:\Gamma_0\rightarrow{\mathbb{R}} \biggm| \|G\|_C:=
\int_{\Gamma_0} |G(\eta)| C^{|\eta|} d\lambda(\eta) <\infty\biggr\}.
\]

In \cite[Proposition~3.1]{FKKZ2010}, it was shown that the mapping $\hL :=K^{-1}LK$
given on $\Bbs(\Ga_0)$ by
\begin{align}\label{Lhat}
(\hL G)(\eta) =& - |\eta| G(\eta) \\& + z
\sum_{\xi\subset\eta}\int_\X e^{-E^\phi(x,\xi)} G(\xi\cup x)e_\la
(e^{-\phi (x - \cdot)}-1,\eta\setminus\xi) dx \nonumber
\end{align}
is a linear operator on $\L_C$ with the dense domain
$\L_{2C}\subset\L_C$. If additionally,
\begin{equation}\label{verysmallparam}
z\leq \min\bigl\{Ce^{-CC_{\phi }} ; 2Ce^{-2CC_{\phi }}\bigr\},
\end{equation}
then $\bigl(\hL , \L_{2C}\bigr)$ is closable linear operator in
$\L_C$ and its closure $\bigl(\hL , D(\hL )\bigr)$ generates a strongly continuous contraction
semigroup $\hT (t)$ on $\L_C$ (see \cite[Theorem~3.8]{FKKZ2010} for details).

Let us set $d\la_C:= C^{|\cdot|} d\la$; then the dual space
$(\L_C)'=\bigl(L^1(\Ga_0, d\la_C)\bigr)'=L^\infty(\Ga_0, d\la_C)$.
The space $(\L_C)'$ is isometrically isomorphic to the Banach space
\[
\K_{C}:=\left\{k:\Ga_{0}\goto\R\,\Bigm| k\cdot C^{-|\cdot|}\in
L^{\infty}(\Ga_{0},\la)\right\}
\]
with the norm
$\|k\|_{\K_C}:=\|C^{-|\cdot|}k(\cdot)\|_{L^{\infty}(\Ga_{0},\la)}$
where the isomorphism is provided by the isometry $R_C$
\begin{equation}\label{isometry}
(\L_C)'\ni k  \longmapsto R_Ck:=k\cdot C^{|\cdot|}\in \K_C.
\end{equation}

In fact, one may consider the duality between the Banach spaces
$\L_C$ and $\K_C$ given by the following expression
\begin{equation}
\llu  G,\,k \rru  := \int_{\Ga_{0}}G\cdot k\, d\la,\quad G\in\L_C, \
k\in \K_C \label{duality}
\end{equation}
with $\lv  \llu   G,k \rru  \rv \leq \|G\|_C \cdot\|k\|_{\K_C}$. It
is clear that $k\in \K_C$ implies $|k(\eta)|\leq \|k\|_{\K_C} \,
C^{|\eta|}$ for $\la$-a.a. $\eta\in\Ga_0$.

Let $\bigl( {\hat{L}} ^{\prime }, D({\hat{L}} ^{\prime })\bigr)$ be
an operator in $(\L_C)^{\prime }$ which is dual to the closed
operator $\bigl( {\hat{L}} , D({\hat{L}} )\bigr)$. We consider also
its image on ${\mathcal{K}}_C$ under the isometry $R_C$, namely, let ${\hat{L}}^{*}=R_C{%
\hat{L}} ^{\prime }R_{C^{-1}}$ with the domain $D({\hat{L}} ^{*})=R_C  D({\hat{L}%
} ^{\prime })$. It was noted in~\cite{FKK2010} that $\hL^\ast$ is
the dual operator to $\hL $ w.r.t. the duality \eqref{duality} and
that for any $k\in D({\hat{L}}^\ast)$
\begin{align}\label{dual-descent}
(\hL^* k)(\eta)=&-\vert \eta \vert k(\eta)\\&+z \sum_{x\in
\eta}e^{-E^\phi (x,\eta\setminus x)} \int_{\Ga_0}e_\la (e^{-\phi (x
- \cdot)}-1,\xi) k((\eta\setminus x)\cup\xi)\,d\la (\xi).\nonumber
\end{align}

Under condition \eqref{verysmallparam}, we consider the adjoint
semigroup $\hT '(t)$ in $(\L_C)'$ and its image $\hT^\ast(t)$ in
$\K_C$. By the general results
from~\cite[Sections~1.2,~1.3]{vNee1992}, the restriction $\hT
^\odot(t)$ of the semigroup $\hT^\ast(t)$ onto its invariant Banach
subspace $\overline{D(\hL^\ast)}$ is a contraction strongly
continuous semigroup. By~\cite[Proposition~3.1]{FKK2010}, for any
$\a\in(0;1)$ we have $\K_{\a C}\subset D(\hL^\ast)$ and, moreover,
by~\cite[Proposition~3.3]{FKK2010}, there exists
$\a_0=\a_0(z,\phi,C)\in (0;1)$ such that for any $\a\in (\a_0;1)$
the set $\K_{\a C}$ will be also a $\hT^\ast(t)$-invariant linear
subspace. As a result, for any $\overline{D(\hL^\ast)}$ the Cauchy
problem in $\K_C$
\begin{equation}
\begin{cases}
\dfrac{\partial}{\partial t} k_t = \hL^\ast k_t \\[2mm] k_t \bigr|_{t=0}
= k_0
\end{cases}
\end{equation}
is well-defined and solvable: $k_t= \hT^\ast (t)k_0= \hT^\odot
(t)k_0\in\overline{D(\hL^\ast)}$; moreover, $k_0\in\K_{\a C}$ implies $k_t\in\K_{\a
C}$.

\section{Vlasov-type scaling}

\subsection{Description of scaling}\label{scalingdescr}

We start from the explanation of the idea of the Vlasov-type
scaling. We want to construct some scaling of the generator $L$,
say, $L_\eps$, $\eps>0$, such that the following scheme holds.
Suppose that we have a semigroup $\hT_\eps(t)$ with generator $\hL
_\eps$ in some $\L_{C_\eps}$. Consider the dual semigroup $\hT
_\eps^\ast(t)$. Let us choose an initial function of the
corresponding Cauchy problem with a big singularity by $\eps$,
namely, $k_0^\e(\eta) \sim \eps^{-|\eta|} r_0(\eta)$, $\eps\goto 0$,
$\eta\in\Ga_0$ with some function $r_0$, independent of $\eps$. Our
first demand to the scaling $L\mapsto L_\eps$ is that the semigroup
$\hT_\eps^\ast(t)$ preserves the order of the singularity:
\begin{equation}\label{ordersing}
(\hT_\eps^\ast(t)k_0^\e)(\eta) \sim \eps^{-|\eta|} r_t(\eta), \quad
\eps\goto 0, \ \ \eta\in\Ga_0.
\end{equation}
And the second one is that the dynamics $r_0 \mapsto r_t$ should
preserve Lebesgue--Poisson exponents, namely, if
$r_0(\eta)=e_\la(\rho_0,\eta)$ then $r_t(\eta)=e_\la(\rho_t,\eta)$
and there exists explicit (nonlinear, in general) differential
equation for $\rho_t$:
\begin{equation}\label{V-eqn-gen}
\dfrac{\partial}{\partial t}\rho_t(x) = \upsilon(\rho_t)(x)
\end{equation}
which we will call the Vlasov-type equation.

Now let us explain an informal way for the realization of this
scheme. Let us consider for any $\eps>0$ the following mapping (cf.
\eqref{isometry}) on functions on $\Ga_0$
\begin{equation}
(R_\eps r)(\eta):=\eps^\n r(\eta).
\end{equation}
This mapping is ``self-dual'' w.r.t. the duality \eqref{duality},
moreover, $R_\eps^{-1}=R_{\eps^{-1}}$. Then we have $k^\e_0\sim
R_{\eps^{-1}} r_0$, and we need $r_t \sim R_\eps
\hT_\eps^\ast(t)k_0^\e \sim  R_\eps \hT _\eps^\ast(t)R_{\eps^{-1}}
r_0$. Therefore, we have to show that for any $t\geq 0$ the operator
family $R_\eps \hT _\eps^\ast(t)R_{\eps^{-1}}$, $\eps>0$ has
limiting (in a proper sense) operator $U(t)$ and
\begin{equation}\label{chaospreserving}
U(t)e_\la(\rho_0)=e_\la(\rho_t).
\end{equation}
But, informally, $\hT^\ast_\eps(t)=\exp{\{t\hL^\ast_\eps\}}$ and
$R_\eps \hT_\eps^\ast(t)R_{\eps^{-1}}=\exp{\{t R_\eps \hL_\eps^\ast
R_{\eps^{-1}} \}}$. Let us consider the ``renormalized'' operator
\begin{equation}\label{renorm_def}
\hLrenadj :=  R_\eps \hL_\eps^\ast R_{\eps^{-1}}.
\end{equation}
In fact, we need that there exists an operator $\hL_V^\ast$ such
that $\exp{\{t R_\eps \hL_\eps^\ast R_{\eps^{-1}} \}}\goto
\exp{\{t\hL_V^\ast\}=:U(t)}$ for which \eqref{chaospreserving}
holds. Therefore, a heuristic way to produce such a scaling
$L\mapsto L_\eps$ is to demand that
\[
\lim_{\eps\goto 0}\left(\dfrac{\partial}{\partial
t}e_\la(\rho_t,\eta)-\hLrenadj e_\la(\rho_t,\eta)\right)=0, \qquad
\eta\in\Ga_0
\]
if only $\rho_t$ is satisfied \eqref{V-eqn-gen}. The point-wise
limit of $\hLrenadj$ will be natural candidate for $\hL_V^\ast$.

Note that \eqref{renorm_def} implies $\hL_\ren=R_{\eps^{-1}}\hL
_\eps R_\eps$. Hence, we will use the following scheme to give
rigorous meaning to all considerations above. We consider, for
a~proper scaling $L_\eps$, the ``renormalized'' operator $\hL_\ren$
and prove that it is a generator of a strongly continuous
contraction semigroup $\hT_\ren(t)$ in $\L_C$. Next, we show that
the formal limit $\hL_V$ of $\hL_\ren$ is also a generator of a
strongly continuous contraction semigroup $\hT_V(t)$ in $\L_C$ also.
Then, we consider the dual semigroups $\hT_\ren^\ast(t)$ and
$\hT_V^\ast(t)$ in the proper Banach subspace of the space $\K_C$.
Finally, we prove that $\hT_\ren^\ast(t)\rightarrow\hT_V^\ast(t)$
strongly on this subspace and explain in which sense $\hT_V^\ast(t)$
satisfies the properties above. Below we try to realize this scheme.

\subsection{Construction and convergence of the evolutions in $\L_C$}\label{subsectiondescevolution}

Let us consider for any $F\in\Fc (\Ga)$, $\eps>0$
\begin{align}
(L_\eps F)(\ga):=&\sum_{x\in\ga} \bigl[F(\ga\setminus x)
-F(\ga)\bigr] \label{genGa-eps}
\\& + \eps^{-1}z \int_{\X} \bigl[F(\ga\cup x)
-F(\ga)\bigr]\exp\bigl\{-\eps E^{\phi}(x,\ga)\bigr\} dx, \quad
\ga\in\Ga.\nonumber
\end{align}
We define also for any $G\in\Bbs(\Ga_0)$, $\eps>0$
\[
\hL_\eps G:=K^{-1}L_\eps K G; \qquad \hL_\ren G:=R_{\eps^{-1}}\hL
_\eps R_\eps G.
\]

Let $\phi$ be integrable function on the whole $\X$, namely,
\begin{equation}\label{integrability}
\beta :=\int_{\X}\phi(x)dx<+\infty.
\end{equation}
We fix this notation for our considerations below.

Then, by the elementary inequality
\begin{equation}
1-e^{- t}\leq  t, \quad t\geq 0  \label{ineq_exp}
\end{equation}%
(which we will use often), $\phi$ will satisfy
\eqref{weak_integrability} and $C_\phi\leq\beta$.

\begin{proposition}
For any $G\in B_{bs}\left( \Ga _{0}\right)$
\begin{equation}
( \hL_{\eps ,\mathrm{ren}}G) \left( \eta \right) =\left(
L_{1}G\right) \left( \eta \right) +\left( L_{2,\eps }G\right) \left(
\eta \right),  \label{expleps}
\end{equation}
where
\begin{align*}
\left( L_{1}G\right) \left( \eta \right) &=-\lv  \eta \rv
G\left( \eta \right) , \\
\left( L_{2,\eps }G\right) \left( \eta \right) &=z\sum_{\xi \subset
\eta }\int_\X e_{\la  }\left( e^{-\eps \phi \left( x-\cdot \right)
},\xi \right) \\&\qquad\times e_{\la }\left( \frac{e^{-\eps \phi
\left( x-\cdot \right) }-1}{\eps },\eta \setminus \xi \right)
G\left( \xi \cup x\right) dx.\nonumber
\end{align*}
Moreover, the expression \eqref{expleps} defines a linear operator
in $\L_C$ with dense domain $\L_{2C}$.
\end{proposition}
\begin{proof}
By \eqref{Lhat}, for any $G\in\Bbs(\Ga_0)$ we have
\begin{align}\label{Lhateps}
(\hL_\eps G)(\eta) =&- |\eta| G(\eta) \\&+ \eps^{-1}z
\sum_{\xi\subset\eta}\int_\X e^{-\eps E^\phi(x,\xi)} G(\xi\cup
x)e_\la (e^{-\eps \phi (x - \cdot)}-1,\eta\setminus\xi) dx
\nonumber.
\end{align}
Then
\begin{align*}
 (\hL_\ren G)(\eta) &= (R_{\eps^{-1}}\hL_\eps R_\eps
G)(\eta)\\ & =-\eps^{-\n} |\eta|\eps^\n  G(\eta) \\&\qquad +
\eps^{-\n}\eps^{-1}z \sum_{\xi\subset\eta}\int_\X e^{-\eps
E^\phi(x,\xi)} \eps^{|\xi\cup x|}G(\xi\cup
x)e_\la (e^{-\eps \phi (x - \cdot)}-1,\eta\setminus\xi) dx\\
& =\left( L_{1}G\right) \left( \eta \right) +\left( L_{2,\eps
}G\right) \left( \eta \right).
\end{align*}
Next, for any $G\in \L_{2C}$ we obtain
\begin{align}
\lV  L_{1}G\rV _{C}&=\int_{\Ga _{0}}\lv  \eta \rv  \lv G\left( \eta
\right) \rv  C^{\lv  \eta \rv  }d\la \left( \eta \right)
\nonumber\\&\leq \int_{\Ga_{0}}2^{\lv \eta \rv  }\lv G\left( \eta
\right) \rv C^{\lv \eta \rv  }d\la \left( \eta \right) =\lV G\rV
_{2C}. \label{bound1}
\end{align}
From \eqref{ineq_exp} and the estimate $e^{-\phi}\leq 1$ we get
\begin{align*}
&\lV  L_{2,\eps }G\rV _{C} \nonumber \\  \leq &z\int_{\Ga
_{0}}\sum_{\xi \subset \eta }\int_\X \lv  G\left( \xi \cup x\right)
\rv  e_{\la }\left( \frac{1-e^{-\eps \phi \left( x-\cdot \right)
}}{\eps },\eta \setminus \xi \right) dxC^{\lv  \eta \rv }d\la
\left( \eta \right) \nonumber\\
\leq  &z\int_{\Ga _{0}}\sum_{\xi \subset \eta }\int_\X \lv  G\left(
\xi \cup x\right) \rv  e_{\la }\left( \phi \left( x-\cdot \right)
,\eta \setminus \xi \right) dxC^{\lv \eta \rv  }d\la  \left( \eta
\right),
\end{align*}
then, by Lemma~\ref{Minlos}, one may continue,
\begin{align*}
\leq &z\int_{\Ga _{0}}\int_{\Ga _{0}}\int_\X \lv  G\left( \xi \cup
x\right) \rv  e_{\la  }\left( \phi \left( x-\cdot \right) ,\eta
\right) dxC^{\lv  \eta \rv  }d\la \left( \eta \right) C^{\lv  \xi
\rv  }d\la  \left( \xi \right)
\end{align*}
and \eqref{LP-exp-mean} yields
\begin{align*}
=&z\exp \left\{ C\beta
\right\} \int_{\Ga _{0}}\int_\X \lv  G\left( \xi \cup x\right) \rv
dxC^{\lv \xi \rv }d\la  \left( \xi \right),
\end{align*}
then, using Lemma~\ref{Minlos} again,
\begin{align}
=&z\exp
\left\{ C\beta \right\} C^{-1}\int_{\Ga _{0}}\lv G\left( \xi \right)
\rv  \cdot\lv  \xi \rv C^{\lv  \xi \rv  }d\la \left( \xi \right)
\nonumber\\ \leq & z\exp \left\{ C\beta \right\} C^{-1}\lV
G\rV_{2C}. \label{bound2}
\end{align}
The estimates \eqref{bound1} and \eqref{bound2} provide the
statement. \end{proof}

\begin{proposition}
Let for any $G\in\Bbs(\Ga_0)$
\begin{equation}
( \hL_V G) \left( \eta \right) :=\lim_{\eps\goto 0} ( \hL_{\eps
,\mathrm{ren}}G) \left( \eta \right)=\left( L_{1}G\right) \left(
\eta \right) +( L_{2}^{V}G) \left( \eta \right) ,\quad \eta\in\Ga_0,
\label{explV}
\end{equation}%
where%
\begin{align*}
( L_{2}^{V}G) \left( \eta \right) =&z\sum_{\xi \subset
\eta }\int_\X G\left( \xi \cup x\right) e_{\la  }\left( -\phi \left(
x-\cdot \right) ,\eta \setminus \xi \right) dx.
\end{align*}%
Then, the expression \eqref{explV} defines a linear operator in
$\L_C$ with dense domain $\L_{2C}$.
\end{proposition}
\begin{proof}
Since, by the definition,
\[
\lV  L_2^VG\rV _{C} \leq z\int_{\Ga _{0}}\sum_{\xi \subset \eta
}\int_\X \lv  G\left( \xi \cup x\right) \rv  e_{\la }\left( \phi
\left( x-\cdot \right) ,\eta \setminus \xi \right) dxC^{\lv \eta \rv
}d\la  \left( \eta \right)
\]
the statement follows from \eqref{bound1} and \eqref{bound2}.
\end{proof}

Let us set (cf. \cite[(3.12)]{FKKZ2010}) for any $\delta \in \left( 0;1\right)$, $\eps>0$, $G\in
B_{bs}\left( \Ga _{0}\right)$, $\eta\in\Ga_0$
\begin{align}
\bigl( \hP_{\delta ,\eps }G\bigr) \left( \eta \right) :=&\sum_{\xi
\subset \eta }\left( 1-\delta \right)^{\lv  \xi \rv }\int_{\Ga
_{0}}\left( z\delta \right)^{\lv \omega
\rv  }G\left( \xi \cup \omega \right)  \label{contreps} \\
&\qquad \times e_{\la  }\left( e^{-\eps E^{\phi }\left( \cdot
,\omega \right) },\xi \right) e_{\la  }\left( \frac{e^{-\eps E^{\phi
}\left( \cdot ,\omega \right) }-1}{\eps },\eta \setminus \xi \right)
d\la \left( \omega \right) . \nonumber
\end{align}
and
\begin{align}
\bigl( \hQ_{\delta }G\bigr) \left( \eta \right) :=&\sum_{\xi \subset
\eta }\left( 1-\delta \right)^{\lv  \xi \rv }\int_{\Ga_{0}}\left(
z\delta \right)^{\lv  \omega \rv  }G\left( \xi
\cup \omega \right)  \label{contrV} \\
&\qquad \times e_{\la  }\left( -E^{\phi }\left( \cdot ,\omega
\right) ,\eta \setminus \xi \right) d\la  \left( \omega \right).
\nonumber
\end{align}
\begin{proposition}\label{contrmaps}
Let
\begin{equation}
ze^ { \beta C } \leq C.  \label{smallz}
\end{equation}%
Then $\hP_{\delta ,\eps }$ and $\hQ_{\delta }$ given by
\eqref{contreps} and \eqref{contrV} are well defined linear
contractions on $\L_C$.
\end{proposition}

\begin{proof}
By \eqref{ineq_exp}, Lemma~\ref{Minlos}, and \eqref{LP-exp-mean}, we
get for any $G\in \L_C$
\begin{align*}
 &\max\Bigl\{ \bigl\Vert  \hP_{\delta ,\eps }G\bigr\Vert _{C} ; \
\bigl\Vert
\hQ_{\delta }G\bigr\Vert _{C} \Bigl\}\\
  \leq  &\int_{\Ga _{0}}\sum_{\xi \subset \eta }\left( 1-\delta
\right)^{\lv  \xi \rv  }\int_{\Ga _{0}}\left( z\delta \right)^{\lv
\omega \rv  }\lv  G\left( \xi \cup \omega \right) \rv e_{\la }\left(
E^{\phi }\left( \cdot ,\omega \right) ,\eta \setminus \xi \right)
d\la  \left( \omega \right) C^{\lv  \eta \rv  }d\la  \left( \eta
\right) \\  = &\int_{\Ga _{0}}\int_{\Ga _{0}}\left( 1-\delta
\right)^{\lv \xi \rv }\int_{\Ga _{0}}\left( z\delta \right)^{\lv
\omega \rv  }\lv G\left( \xi \cup \omega \right) \rv e_{\la }\left(
E^{\phi }\left( \cdot ,\omega \right) ,\eta \right) d\la \left(
\omega \right) C^{\lv \eta \rv  }d\la \left( \eta \right) C^{\lv \xi
\rv }d\la \left( \xi
\right) \\
 = &\int_{\Ga _{0}}\int_{\Ga _{0}}\left( 1-\delta
\right)^{\lv \xi \rv }\left( z\delta \right)^{\lv  \omega \rv  }\lv
G\left( \xi \cup \omega \right) \rv  \exp \left\{ C\beta \lv \omega
\rv \right\} d\la \left( \omega \right) C^{\lv  \xi \rv  }d\la
\left( \xi \right) \\  = &\int_{\Ga _{0}}\int_{\Ga _{0}}\left(
1-\delta \right)^{\lv \xi \rv }\lv  G\left( \xi \cup \omega \right)
\rv \left( z\delta \exp \left\{ C\beta \right\} C^{-1}\right)^{\lv
\omega \rv }C^{\lv \omega \rv  }C^{\lv  \xi \rv }d\la  \left( \xi
\right) d\la \left( \omega \right) \\
 = &\int_{\Ga _{0}}\lv  G\left( \xi \right) \rv  \left( 1-\delta
+z\delta \exp \left\{ C\beta \right\} C^{-1}\right)^{\lv \xi \rv
}C^{\lv  \xi \rv  }d\la \left( \xi \right) \leq \lV G\rV_{C},
\end{align*}
that proves the contraction property; then, in particular,
\[
\bigl( \hP_{\delta ,\eps }G\bigr) \left( \eta \right)<+\infty,
\qquad \bigl( \hQ_{\delta }G\bigr) \left( \eta \right)<+\infty
\]
for $\la$-a.a. $\eta\in\Ga_0$. \end{proof}

Now let us construct the approximations for the operators $\hL_V $
and $\hL_\ren$.

\begin{proposition}\label{apprgenGl}
Let for $\delta \in \left( 0;1\right) $%
\[
\hL_{\delta , V }:=\frac{1}{\delta }\bigl( \hQ_{\delta }-1\bigr) ;~~~%
\hL_{\delta ,\eps }:=\frac{1}{\delta }\bigl( \hP_{\delta ,\eps
}-1\bigr) , \ \eps>0.
\]%
Let \eqref{smallz} holds, then
\[
\lV  \bigl( \hL_{\delta , V }-\hL_V \bigr) G\rV_{C}<3\delta \lV
G\rV_{2C}
\]%
and for any $\eps >0$%
\[
\lV  \bigl( \hL_{\delta ,\eps }-\hL_{\eps ,%
\mathrm{ren}}\bigr) G\rV  _{C}\leq 3\delta \lV G\rV_{2C}.
\]
\end{proposition}

\begin{proof}
Let us denote%
\begin{align*}
\bigl( \hQ_{\delta }^{\left( 0\right) }G\bigr) \left( \eta \right)
:=&\sum_{\xi \subset \eta }\left( 1-\delta \right)^{\lv \xi \rv
}G\left( \xi \right) 0^{\lv  \eta \setminus \xi \rv }=\left(
1-\delta \right)^{\lv  \eta \rv  }G\left(
\eta \right) , \\
\bigl( \hQ_{\delta }^{\left( 1\right) }G\bigr) \left( \eta \right)
:=&z\delta \sum_{\xi \subset \eta }\left( 1-\delta \right)^{\lv \xi
\rv  }\int_\X G\left( \xi \cup x\right) e_{\la }\left( -\phi \left(
x-\cdot \right) ,\eta \setminus \xi \right) dx,
\end{align*}%
and
\[
\hQ_{\delta }^{\left( \geq 2\right) }:=\hQ_{\delta }-\bigl( \hQ
_{\delta }^{\left( 0\right) }+\hQ_{\delta }^{\left( 1\right) }\bigr)
.
\]%
Clearly, we have
\begin{align*}
\lV  \bigl( \hL_{\delta , V }-\hL_V \bigr) G\rV _{C} \leq &\lV
\frac{1}{\delta }\bigl( \hQ_{\delta }^{\left( 0\right)
}-1\bigr) G-L_{1}G\rV _{C} \\
&+\lV   \frac{1}{\delta }\hQ_{\delta }^{\left( 1\right)
}G-L_{2}^{V} G\rV _{C}+\lV  \frac{1}{\delta }\hQ %
_{\delta }^{\left( \geq 2\right) }G\rV _{C}.
\end{align*}%
It follows from the simple inequality
\begin{equation}\label{usefulineq}
0\leq n-\frac{1-\left( 1-\delta \right)^{n}}{\delta }<\delta \cdot
2^{n},\ n\in \N,\ \delta >0,
\end{equation}
that
\[
\lV  \frac{1}{\delta }\bigl( \hQ_{\delta }^{\left( 0\right)
}-1\bigr) G-L_{1}G\rV _{C}=\lV  \frac{1}{\delta }\left( \left(
1-\delta \right)^{\lv  \cdot \rv }-1\right) G+\lv \cdot \rv G\rV
_{C}<\delta \lV  G\rV _{2C}
\]%
and%
\begin{align*}
 &\lV  \frac{1}{\delta }\hQ_{\delta }^{\left( 1\right)
}G-L_{2}^{V}G\rV _{C} \\
 \leq & z\int_{\Ga _{0}}\lv  \sum_{\xi \subset \eta }\left[
\left( 1-\delta \right)^{\lv  \xi \rv  }-1\right] \int_\X G\left(
\xi \cup x\right) e_{\la  }\left( -\phi \left( x-\cdot \right) ,\eta
\setminus \xi \right) dx\rv  C^{\lv  \eta
\rv  }d\la  \left( \eta \right) \\
 \leq  & z\int_{\Ga _{0}}\int_{\Ga _{0}}\left[ 1-\left(
1-\delta \right)^{\lv  \xi \rv  }\right] \int_\X \lv  G\left( \xi
\cup x\right) \rv  e_{\la  }\left( \phi \left( x-\cdot \right) ,\eta
\right) dxC^{\lv  \eta \rv  }C^{\lv  \xi
\rv  }d\la  \left( \eta \right) d\la  \left( \xi \right) \\
 = &  z\exp \left\{ C\beta \right\} \int_{\Ga _{0}}\left[ 1-\left(
1-\delta \right)^{\lv  \xi \rv  }\right] \int_\X \lv G\left( \xi
\cup x\right) \rv  dxC^{\lv  \xi \rv  }d\la  \left( \xi \right)
\\\leq
 & \delta z\exp \left\{ C\beta \right\}
\int_{\Ga_{0}}\lv  \xi \rv  \int_\X \lv  G\left( \xi \cup x\right)
\rv  dxC^{\lv  \xi \rv  }d\la  \left( \xi \right) \\
 = &\delta z\exp \left\{ C\beta \right\}
C^{-1}\int_{\Ga_{0}}\lv \xi \rv  \left( \lv  \xi \rv  -1\right) \lv
G\left( \xi \right) \rv C^{\lv  \xi \rv  }d\la \left( \xi \right)
<\delta \lV G\rV_{2C},
\end{align*}%
since, $n\left( n-1\right) \leq 2^{n}$, $n\in \N$. And, if we denote
\[
\Ga _{0}^{\left( \geq 2\right) }:=\bigsqcup\limits_{n\geq 2}\Ga
_{0}^{\left( n\right) },
\]
we obtain
\begin{align*}
&\lV  \frac{1}{\delta }\hQ_{\delta }^{\left( \geq 2\right)
}G\rV _{C} \\
\leq &\frac{1}{\delta }\int_{\Ga _{0}}\sum_{\xi \subset \eta }\left(
1-\delta \right)^{\lv  \xi \rv  }\int_{\Ga_{0}^{\left( \geq 2\right)
}}\left( z\delta \right)^{\lv \omega \rv
}\lv  G\left( \xi \cup \omega \right) \rv  \\
&\qquad \times e_{\la  }\left( E^{\phi }\left( \cdot ,\omega \right)
,\eta \setminus \xi \right) d\la  \left( \omega \right) C^{\lv  \eta
\rv  }d\la  \left( \eta \right) \\
=&\frac{1}{\delta }\int_{\Ga _{0}}\int_{\Ga _{0}}\left( 1-\delta
\right)^{\lv  \xi \rv  }\int_{\Ga _{0}^{\left( \geq 2\right)
}}\left( z\delta \right)^{\lv  \omega \rv
}\lv  G\left( \xi \cup \omega \right) \rv  \\
&\qquad \times e_{\la  }\left( E^{\phi }\left( \cdot ,\omega \right)
,\eta \right) d\la  \left( \omega \right) C^{\lv  \eta \rv  }d\la
\left( \eta \right) C^{\lv  \xi \rv  }d\la \left(
\xi \right) \\
\leq &\delta \int_{\Ga _{0}}\left( 1-\delta \right)^{\lv \xi \rv
}\int_{\Ga _{0}}\left( z\exp \left\{ C\beta \right\} \right)^{\lv
\omega \rv  }\lv  G\left( \xi \cup \omega \right) \rv  d\la \left(
\omega \right) C^{\lv  \xi \rv
}d\la  \left( \xi \right) \\
=&\delta \int_{\Ga _{0}}\left( C-\delta C+z\exp \left\{ C\beta
\right\} \right)^{\lv  \xi \rv  }\lv  G\left( \xi \right)
\rv  d\la  \left( \xi \right) \\
\leq &\delta \int_{\Ga _{0}}\left( 2C-\delta C\right)^{\lv \xi \rv
}\lv  G\left( \xi \right) \rv  d\la \left( \xi \right) <\delta \lV
G\rV _{2C}.
\end{align*}

The same considerations may be done for $\hP_{\delta ,\eps }$.
Namely, let%
\begin{align*}
\bigl( \hP_{\delta ,\eps }^{\left( 0\right) }G\bigr) \left( \eta
\right) :=&\sum_{\xi \subset \eta }\left( 1-\delta \right)^{\lv \xi
\rv  }G\left( \xi \right) 1^{\lv \xi \rv }0^{\lv \eta \setminus \xi
\rv  }=\left( 1-\delta \right)
^{\lv  \eta \rv  }G\left( \eta \right) , \\
\bigl( \hP_{\delta ,\eps }^{\left( 1\right) }G\bigr) \left( \eta
\right) :=&z\delta \sum_{\xi \subset \eta }\left( 1-\delta \right)
^{\lv  \xi \rv  }\int_\X G\left( \xi \cup x\right)
\\
&\qquad \times e_{\la  }\left( e^{-\eps \phi \left( x-\cdot \right)
},\xi \right) e_{\la  }\left( \frac{e^{-\eps \phi \left( x-\cdot
\right) }-1}{\eps },\eta \setminus \xi \right) dx,
\end{align*}%
and
\[
\hP_{\delta ,\eps }^{\left( \geq 2\right) }:=\hP_{\delta ,\eps
}-\bigl( \hP_{\delta ,\eps }^{\left( 0\right) }+\hP_{\delta ,\eps
}^{\left( 1\right) }\bigr) .
\]%
Then%
\[
\lV  \frac{1}{\delta }\bigl( \hP_{\delta ,\eps }^{\left(
0\right) }-1\bigr) G-L_{1}G\rV _{C}=\lV  \frac{1}{\delta }%
\bigl( \hQ_{\delta }^{\left( 0\right) }-1\bigr) G-L_{1}G\rV
_{C}<\delta \lV  G\rV _{2C},
\]%
next, by \eqref{ineq_exp}, \eqref{smallz} and Lemma~\ref{Minlos},%
\begin{align}
 & \lV  \frac{1}{\delta }\hP_{\delta ,\eps }^{\left(
1\right) }G-L_{2,\eps }G\rV _{C} \nonumber\\
 \leq & z\int_{\Ga _{0}}\int_{\Ga _{0}}\left[ 1-\left(
1-\delta \right)^{\lv  \xi \rv  }\right] \int_\X \lv  G\left( \xi
\cup x\right) \rv  e_{\la  }\left( \phi \left( x-\cdot \right) ,\eta
\right) dxC^{\lv  \eta \rv  }C^{\lv  \xi
\rv  }d\la  \left( \eta \right) d\la  \left( \xi \right) \nonumber\\
 < & \lV  G\rV _{2C}\cdot \delta e^{C\beta} C^{-1} z \leq
\delta \lV G\rV _{2C},
\end{align}%
and, finally,%
\begin{align*}
\lV  \frac{1}{\delta }\hP_{\delta ,\eps }^{\left( \geq 2\right)
}G\rV _{C} \leq &\frac{1}{\delta }\int_{\Ga _{0}}\sum_{\xi \subset
\eta }\left( 1-\delta \right)^{\lv  \xi \rv  }\int_{\Ga_{0}^{\left(
\geq 2\right) }}\left( z\delta \right)^{\lv \omega \rv
}\lv  G\left( \xi \cup \omega \right) \rv  \\
&\qquad \times \lv  e_{\la  }\left( E^{\phi }\left( \cdot ,\omega
\right) ,\eta \setminus \xi \right) \rv  d\la  \left( \omega \right)
C^{\lv \eta \rv  }d\la  \left( \eta \right) <\delta \lV G\rV_{2C}.
\end{align*}
Combining all these inequalities we obtain the assertion.
\end{proof}

We will need the following results in the sequel.

\begin{lemma}[{\cite[Corollary 3.8]{EK1986}}] \label{EK_res}
Let $A$ be a linear operator on a Banach space $L$ with $D\left(
A\right) $ dense in $L$, and let $|\!|\!| \cdot |\!|\!|$ be a norm
on $D\left( A\right) $ with respect to which $D\left( A\right) $ is
a Banach space. For $n\in \mathbb{N}$ let $T_{n}$ be a linear
$\left\Vert \cdot \right\Vert $-contraction on $L$ such that
$T_{n}:D\left( A\right) \rightarrow D\left( A\right) $, and define
$A_{n}=n\left( T_{n}-1\right) $. Suppose there exist $\omega \geq 0$
and a sequence $\left\{ \varepsilon _{n}\right\} \subset \left(
0;+\infty \right) $ tending to zero such that
for $n\in \mathbb{N}$%
\begin{equation}\label{approperEK}
\left\Vert \left( A_{n}-A\right) f\right\Vert \leq \varepsilon
_{n}|\!|\!| f |\!|\!|,~f\in D\left( A\right)
\end{equation}
and
\begin{equation}\label{psevdocontr}
|\!|\!| T_{n}\upharpoonright _{D(A)} |\!|\!| \leq 1+\frac{\omega
}{n}.
\end{equation}
Then $A$ is closable and the closure of $A$ generates a strongly
continuous contraction semigroup on $L$.
\end{lemma}

\begin{lemma}[cf. {\cite[Theorem 6.5]{EK1986}}] \label{EK_res-conv}
Let $L, L_n$, $n\in\N$ be Banach spaces, and $p_n: L\rightarrow L_n$
be bounded linear transformation, such that $\sup_n \|p_n\|<\infty
$. For any $n\in\N$, let $T_n$ be a linear contraction on $L_n$, let
$\eps_n>0$ be such that $\lim_{n\rightarrow \infty} \eps_n =0$, and
put $A_n=\eps_n^{-1}(T_n - \1)$. Let $T_t$ be a strongly continuous
contraction semigroup on $L$ with generator $A$ and let $D$ be a
core for $A$. Then the following are equivalent:
\begin{enumerate}
\item For each $f\in L$, $T_n^{[t/\eps_n]} p_n f\rightarrow p_n
T_t f$ in $L_n$ for all $t\geq0$ uniformly on bounded intervals.
Here and below $[\,\cdot\,\,]$ mean the entire part of a real
number.

\item For each $f\in D$, there exists $f_n\in L_n$ for each
$n\in\N$ such that $f_n \rightarrow p_n f$ and $A_n f_n \rightarrow
p_n Af$ in $L_n$.
\end{enumerate}
\end{lemma}

\begin{lemma}\label{powersofcontractions}
Let $X$ be a Banach space with a norm $\|\cdot\|_X$; $A$ and $B$ be
linear contraction mappings on $X$. Let $Y$ with a norm
$\|\cdot\|_Y$ be a Banach subspace of $X$ such that  $Y$ is
invariant w.r.t. $B$. Suppose that the restriction of $B$ on $Y$ is
also a  contraction w.r.t. $\|\cdot\|_Y$. Suppose also that there
exists $c>0$ such that for any $f\in Y$
\begin{equation}\label{estY}
\| Af-Bf\|_X \leq c\|f\|_Y.
\end{equation}
Then for any $m\in\N$ and for any $f\in Y$
\begin{equation}\label{estYm}
\| A^m f-B^mf\|_X \leq cm\|f\|_Y.
\end{equation}
\end{lemma}
\begin{proof}
For any $f\in Y$, $m\geq2$ we have
\begin{align*}
&\|A^m f -B^m f\|_X\\\leq&\|A^m f -AB^{m-1}f\|_X + \|AB^{m-1}f-B^m
f\|_X\\\leq&\|A\|\cdot\|A^{m-1} f -B^{m-1}f\|_X+\|(A-B)B^{m-1}f\|_X
\\ \intertext{(where $\|A\|$ means the norm of the operator $A$ on $X$); since
$\|A\|\leq1$ and $B^{m-1}f\in Y$, condition \eqref{estY} yields}
\leq&\|A^{m-1} f -B^{m-1}f\|_X + c\|B^{m-1}f\|_Y,
\\\intertext{
but, $B$ is a contraction on $Y$, therefore, one get}
\leq&\|A^{m-1} f -B^{m-1}f\|_X
+ c\|f\|_Y,
\end{align*}
that gives \eqref{estYm} by induction principle. \end{proof}

And now one can construct the corresponding semigroups rigorously.

\begin{proposition}\label{descsemigroupexist}
 Let
\begin{equation}
z\leq \min \left\{ Ce^{-C\beta },2Ce^{-2C\beta }\right\} .
\label{verysmallz}
\end{equation}%
Then, $\bigl( \hL_V ,\L_{2C}\bigr) $ and $\bigl( \hL %
_{\eps ,\mathrm{ren}},\L_{2C}\bigr) $ are closable linear operators
in $\L_C$ and their closures $\bigl( \hL_V , D(\hL_V) \bigr) $ and $\bigl( \hL
_{\eps ,\mathrm{ren}}, D(\hL
_{\eps ,\mathrm{ren}}) \bigr) $ generate strongly continuous
contraction semigroups $\hT_{V}(t)$ and $\hT_{\ren }(t)$ on $\L_C$,
respectively. Moreover, for any $G\in \L_C$, $\eps >0$
\begin{equation}
\hQ_{\frac{1}{n}}^{[nt] }G\goto \hT_{V}(t)G,\qquad \hP
_{\frac{1}{n},\eps }^{[nt] }G\goto \hT_{\ren }(t)G,\qquad n\goto
\infty ~  \label{approx}
\end{equation}%
for any $t\geq 0$ uniformly on bounded intervals.
\end{proposition}
\begin{proof}
Note that \eqref{verysmallz} provides that $%
\hQ_{\delta }$ and $\hP_{\delta ,\eps }$ are also contractions on $%
\L_{2C}$. Then the  first part of the statement follows from
Lemma~\ref{EK_res}. Therefore, $\L_{2C}$ will be a core for the
generators and, by Lemma~\ref{EK_res-conv}, we obtain the
convergence~\eqref{approx}. \end{proof}


The definition \eqref{explV} of $\hL_V$ together with
Proposition~\ref{descsemigroupexist} allow us to expect that the
semigroup $\hT_{\ren }(t)$ converges to $\hT_{V}(t)$ in a proper
sense. The next theorem improve this statement. However, this result
is not crucial in the context of the our paper. Moreover, its proof
is quite technical and, on the other hand, is very similar to the
proof of the main Theorem~\ref{maintheorem} concerning the dual
semigroups. Hence, we give the sketch of the proof only.

\begin{theorem}\label{descsemigroupconv}
 Let \eqref{verysmallz} holds and suppose that $\bar{\phi}:=\sup_\X \phi \left( x\right) <+\infty $.
 Then for any $G\in \L_{2C}$
\[
\lV  \hT_{\ren }(t)G-\hT_{V}(t)G\rV _{C}\leq \eps
t\,\bar{\phi}\left( 1+\beta \right) \lV  G\rV _{2C}
\]%
for any $t\geq 0$, $\eps >0$. In particular, it means that $\hT
_{\ren }(t) G\goto \hT_{V}(t)G$ in $\L_C$ as $\eps \goto 0$ for any
$t\geq 0$ uniformly on bounded intervals.
\end{theorem}

\begin{proof}
By the triangle inequality,%
\begin{align}
\lV  \hT_{\ren }(t)G-\hT_{V}(t)G\rV _{C} \leq
&\lV  \hT_{\ren }(t)G-\hP_{\frac{1}{n},\eps }^{%
[nt] }G\rV _{C}  \label{triangleineq} \\
&+\lV  \hP_{\frac{1}{n},\eps }^{[nt] }G-\hQ %
_{\frac{1}{n}}^{[nt] }G\rV _{C }+\lV  \hQ %
_{\frac{1}{n}}^{[nt] }G-\hT_{V}(t)G\rV _{C}. \nonumber
\end{align}%
By \eqref{approx}, the first and third norms in the r.h.s. of
\eqref{triangleineq} are tend to $0$ as $n\rightarrow\infty$. Next,
in a similar way as for the proof of \eqref{dopex} one can show that
for any $G\in \L_{2C}$
\begin{equation}
\lV  \hP_{\frac{1}{n},\eps }G-\hQ_{\frac{1}{n}%
}G\rV _{C }\leq \frac{1}{n} \,
\eps \, \bar{\phi}\left( 1+\beta \right) \lV  G\rV _{2C}. \label{mainestq}
\end{equation}%
By Proposition~\ref{contrmaps} and condition \eqref{verysmallz}, the
subspace $\L_{2C}$ is $\hQ_{\frac{1}{n}}$-invariant, hence, by
Lemma~\ref{powersofcontractions}, we obtain
\begin{align*}
\lV  \hP_{\frac{1}{n},\eps }^{[nt] }G-\hQ_{%
\frac{1}{n}}^{[nt] }G\rV _{C }&\leq [nt]
\frac{\bar{\phi}\left( 1+\beta \right)}{n}\eps \lV  G\rV _{2C}\\ &< \bar{\phi}\left( 1+\beta \right)\left( t+\frac{1}{n}%
\right) \eps \lV  G\rV _{2C},
\end{align*}%
that fulfilled the first assertion. And, clearly, $\L_{2C}$ is a
dense subspace of $\L_C$. \end{proof}

\subsection{Convergence of the evolutions in $\K_C$}

Let $\eps>0$ be given. Let $\bigl( \hL '_\ren, D(\hL '_\ren)\bigr)$
and $\bigl( \hL '_V, D(\hL '_V)\bigr)$ be dual operators to the
closed operators $\bigl( \hL _\ren, D(\hL _\ren)\bigr)$ and $\bigl(
\hL _V, D(\hL _V)\bigr)$ in the Banach space $(\L_C)'$. Let the
operators $\bigl( \hLrenadj, D(\hLrenadj)\bigr)$ and $\bigl( \hL
^\ast_V, D(\hL ^\ast_V)\bigr)$ be their images in the space $\K_C$
under the isometry \eqref{isometry}. Our aim is to transfer the
previous results onto $\ast$-objects. However, similarly to the case
of the operator $\hL^\ast$ (see Subsection~\ref{dualconstraction}),
the space $\K_C$ is too big. The reason is that the dual semigroup
in a non-reflexive case (namely, $L^1$ case) will not be a strongly
continuous semigroup on the whole dual space. Hence, we consider
some Banach subspace of $\K_C$ which will be useful for the strong
continuity property.

\begin{proposition}\label{adjgen}
For any $\a\in(0;1)$, $\eps>0$, and $k\in\K_\aC$ we have that
\begin{equation}\label{incl}
\bigl\{\hL_{\ren}^\ast k, \,
\hL_{V}^\ast k\bigr\}\subset\K_C.
\end{equation}
Moreover, for any $k\in \K_\aC$
\begin{align}
( \hL_{\ren}^\ast k) \left( \eta \right)  =&-\lv  \eta \rv k\left(
\eta \right) \label{genrenadj}\\& +z\sum_{x\in\eta}\int_{\Ga_0}
e_{\la  }\left( e^{-\eps \phi \left( x-\cdot \right) },\eta\setminus
x \right) \nonumber\\&\qquad\qquad \times e_{\la }\left(
\frac{e^{-\eps \phi \left( x-\cdot \right) }-1}{\eps },\xi   \right)
k\left( \xi \cup \eta\setminus x\right) d\la(\xi)\nonumber
\end{align}
and
\begin{align}
(\hL_{V}^\ast k) \left( \eta \right)  =&-\lv  \eta \rv k\left( \eta
\right) \label{genVladj} \\ &+z\sum_{x\in\eta}\int_{\Ga_0} e_{\la
}\left(- \phi \left( x-\cdot \right),\xi   \right) k\left( \xi \cup
\eta\setminus x\right) d\la(\xi).\nonumber
\end{align}
\end{proposition}
\begin{proof}
By Lemma~\ref{Minlos}, for any $G\in\Bbs(\Ga_0)$ we have
\begin{align*}
&\int_{\Ga_0}\sum_{\xi \subset
\eta }\int_\X e_{\la  }\left( e^{-\eps \phi \left( x-\cdot \right)
},\xi \right) e_{\la }\left( \frac{e^{-\eps
\phi \left( x-\cdot \right) }-1}{\eps },\eta \setminus \xi \right)
\\&\qquad\qquad\times G\left( \xi \cup x\right) dx k(\eta)d\la(\eta)\\&=
\int_{\Ga_0}\int_{\Ga_0}\int_\X e_{\la  }\left( e^{-\eps \phi \left( x-\cdot \right)
},\xi \right) e_{\la }\left( \frac{e^{-\eps
\phi \left( x-\cdot \right) }-1}{\eps },\eta \right)
\\ &\qquad\qquad\times G\left( \xi \cup x\right) dx k(\eta\cup \xi)d\la(\xi)d\la(\eta)\\
&=\int_{\Ga_0}\int_{\Ga_0}\sum_{x\in\xi} e_{\la  }\left( e^{-\eps
\phi \left( x-\cdot \right) },\xi \setminus x\right) e_{\la }\left(
\frac{e^{-\eps \phi \left( x-\cdot \right) }-1}{\eps },\eta \right)
\\ &\qquad\qquad\times G\left( \xi \right) dx k(\eta\cup \xi\setminus x)d\la(\xi)d\la(\eta),
\end{align*}
that implies \eqref{genrenadj}. The equality \eqref{genVladj} may be
obtained in the same way or just as a point-wise limit of
\eqref{genrenadj} as $\eps\goto 0$.

The inclusion \eqref{incl} follows from the estimate ($k\in\K_\aC$)
\begin{align*}
&\frac{1}{C^\n}\sum_{x\in\eta}\int_{\Ga_0} e_{\la  }\left( e^{-\eps
\phi \left( x-\cdot \right) },\eta\setminus x \right) e_{\la }\left(
\left\vert\frac{e^{-\eps \phi \left( x-\cdot \right) }-1}{\eps
}\right\vert,\xi   \right) \bigl\vert k\left( \xi \cup \eta\setminus
x\right) \bigr\vert d\la(\xi)\\ \leq &
\frac{\|k\|_{\K_\aC}}{C^\n}\sum_{x\in\eta}\int_{\Ga_0} e_{\la
}\left( \phi \left( x-\cdot \right),\xi   \right) (\a C)^{|\xi \cup
\eta\setminus x|} d\la(\xi)\\= &\frac{\|k\|_{\K_\aC}\cdot \exp\{\a
C\beta\}}{\a C} |\eta| \a^{\n}\leq \frac{\|k\|_{\K_\aC}\cdot
\exp\{\a C\beta\}}{\a C} \cdot\frac{-1}{e\ln\a},
\end{align*}
where we used that $x\a^x \leq -\dfrac{1}{e\ln \a }$ for any $\a
\in(0;1)$ and $x\geq0$; and the similar estimates for
\begin{equation}
\frac{1}{C^\n} \n \left\vert k(\eta) \right\vert, \quad
\frac{1}{C^\n}\sum_{x\in\eta}\int_{\Ga_0} e_{\la }\left( \phi \left(
x-\cdot \right) ,\xi   \right) \bigl\vert k\left( \xi \cup
\eta\setminus x\right) \bigr\vert d\la(\xi). \qedhere
\end{equation}
\end{proof}

Let now \eqref{verysmallz} holds. By
Proposition~\ref{descsemigroupexist}, there exist strongly
continuous contraction semigroups $\hT _\ren(t)$ and $\hT_V(t)$ on
$\L_C$. Then the corresponding dual semigroups
 $\hT '_\ren(t)$ and $\hT '_V(t)$ act in the space $(\L_C)'$.
Let us denote by $\hT ^\ast_\ren(t)$ and $\hT^\ast_V(t)$ their
corresponding images in $\K_C$ under the isometry \eqref{isometry}.

Proposition~\ref{adjgen} yields that for any $\a\in(0;1)$ the following inclusion holds
\begin{equation}\label{inclusion}
\overline{\K_\aC} \subset \Bigl(\bigcap_{\eps>0} \overline{D(\hL
_\ren^\ast)}\,\Bigr)\bigcap \overline{D(\hL_V^\ast)}
\end{equation}
(all closures are in $\K_C$; in particular, $\overline{\K_\aC}$
is a Banach space with norm \mbox{$\|\cdot\|_C$}).

Moreover, by, e.g., \cite[Sections~1.2,~1.3]{vNee1992} or \cite[Subsection~II.2.5]{EN2000}, for any $\eps>0$ the restrictions
$\hT^\odot_\ren(t)$ and $\hT^\odot_V(t)$ of $\hT^\ast_\ren(t)$ and
$\hT^\ast_V(t)$ onto $\overline{D(\hLrenadj)}$ and
$\overline{D(\hL_V^\ast)}$, correspondingly, are strongly continuous
semigroups; their generators $\hL_\ren^\odot$ and $\hL_V^\odot$ are
the parts of $\hLrenadj$ and $\hL_V^\ast$, correspondingly. Namely,
\begin{align*}
D(\hL^\odot_\ren)&=\bigl\{k\in D(\hL ^\ast_\ren) \bigm| \hL
_\ren^\ast k\in \overline{D(\hLrenadj)} \bigr\}, \\
D(\hL^\odot_V)&=\bigl\{k\in D(\hL^\ast_V) \bigm| \hL_V^\ast k\in
\overline{D(\hL_V^\ast)} \bigr\},
\end{align*}
and
\begin{align*}
\hLrenadj k &=\hL^\odot_\ren
k, \qquad k\in D(\hL_\ren^\odot),\\
\hL^\ast_V k &=\hL^\odot_V k,  \qquad k\in D(\hL_V^\odot).
\end{align*}

\begin{proposition}\label{sun-inv}
Assume that, as before,
\begin{equation}
z\leq \min \left\{ Ce^{-C\beta },2Ce^{-2C\beta }\right\} .
\label{verysmallz-2}
\end{equation}%
If $C \beta=\ln2$ we suppose additionally that $z<\frac{C}{2}$.
Then, there exists $\a_1=\a_1(z,\beta,C)\in(0;1)$ such that for any
$\a\in(\a_1;1)$ the space $\KK$ will be $\hT^\odot_V(t)$- and $\hT
^\odot_\ren(t)$-invariant, $\eps>0$.
\end{proposition}
\begin{proof}
The proof is fully analogous to that of
\cite[Proposition~3.3]{FKK2010}. For readers convince we explain it
in details.

By \eqref{verysmallz-2}, $z\beta\leq\min\{C\beta e^{-C\beta},
2C\beta e^{-2C\beta}\}$. Note that the function $f(x)=x e^{-x}$,
$x\geq 0$ is increasing on $(0; 1)$ from $0$ to $e^{-1}$ and it is
asymptotically decreasing on $(1;+\infty)$ from $e^{-1}$ to $0$.
Therefore, if $C\beta e^{-C\beta}\neq 2C\beta e^{-2C\beta}$ then
\eqref{verysmallz-2} with necessity implies $z \beta < e^{-1}$.
Otherwise, if $C\beta=\ln 2$ then the condition $2z<C$ implies
$z\beta<\frac{C\beta}{2}=C\beta e^{-C\beta}=2C\beta e^{-2C\beta}$,
and, again, $z\beta<e^{-1}$. As a result, the equation $f(x)=z
\beta$ has exactly two roots, say, $0<x_1<1<x_2<+\infty$. Therefore,
$x_1< C \beta < 2 C \beta < x_2$.

If $C\beta>1$ then we  set $\a_1:=\max\left\{\frac{1}{2};\frac{1}{C
\beta};\frac{1}{C}\right\}<1$. This yields $2\a C \beta > C
\beta$ and $\a C\beta >1>x_{1}$. If $x_{1}<C\beta \leq 1$ then we
set $\a_1:=\max\left\{\frac{1}{2};\frac{x_{1}}{C
\beta};\frac{1}{C}\right\}<1$ that gives $2\a C \beta > C \beta$
and $\a C\beta >x_{1}$. As a result,
\begin{equation}\label{ineq-alpha}
x_1<\a C\beta  < C\beta <2 \a C \beta < 2 C \beta < x_2
\end{equation}
and $1<\a C<C<2\a C<2C$. The last inequality shows that $\L_{2C
}\subset\L_{2\a C}\subset \L_C\subset \L_{\a C}$.

By \eqref{ineq-alpha}, $z\beta<\min\{f(\aC\beta),f(2\aC\beta)\}$,
hence, $z<\min\{\aC e^{-\aC\beta}, 2\aC
e^{-2\aC\beta}\}$. Then, analogously to Proposition~\ref{descsemigroupexist}, we obtain that the operators
$\bigl( \hL_V ,\L_{2\aC}\bigr) $ and $\bigl( \hL %
_{\eps ,\mathrm{ren}},\L_{2\aC}\bigr) $
are closable in $\L_{\a C}$ and their closures are generators of
contraction semigroups, say, $\hT_{\a,V} (t)$ and $\hT_{\a,\ren} (t)$ on $\L_{\a C}$, correspondingly.

It is easy to see, that $\hT_{\a,V} (t) G= \hT_V (t) G$ and $\hT_{\a,\ren} (t) G= \hT_\ren (t) G$ for any
$G\in\L_C$. Indeed, since the contraction mappings $\hQ_\delta$ and $\hP_{\delta,\eps}$, $\delta,\eps>0$
 do not depend on $\a$, we obtain, by Proposition~\ref{descsemigroupexist}, that for any
$G\in\L_{C }\subset\L_{\a C}$ we have that $\hT_V (t)G\in\L_{C
}\subset\L_{\a C}$ and $\hT_{\a,V} (t) G\in\L_{\a C}$ and
\begin{align*}
&\| \hT_V (t)G-\hT_{\a,V} (t) G\|_{\a C}\\
\leq &\Bigl\| \hT_V (t)G-\hQ
_\delta^{\left[\frac{t}{\delta}\right]}G\Bigr\|_{\a C} + \Bigl\| \hT_{\a,V}
(t) G-\hQ_\delta^{\left[\frac{t}{\delta}\right]}G\Bigr\|_{\a C}\\ \leq&
\Bigl\| \hT_V (t)G-\hQ_\delta^{\left[\frac{t}{\delta}\right]}G\Bigr\|_{ C} +
\Bigl\| \hT_{\a,V} (t)
G-\hQ_\delta^{\left[\frac{t}{\delta}\right]}G\Bigr\|_{\a C}\goto0,
\end{align*}
as $\delta\goto 0$. Therefore, $\hT_V (t)G=\hT_{\a,V} (t) G$ in $\L_{\a
C}$ (recall that $G\in\L_C$) that yields $ \hT_V (t)G(\eta)=\hT_{\a,V} (t)
G(\eta)$ for $\la$-a.a. $\eta\in\Ga_0$ and, therefore, $\hT_V
(t)G=\hT_{\a,V} (t) G$ in $\L_{C}$.

Note that for any $G\in\L_C\subset\L_{\a C}$ and for any $k\in
\K_{\a C}\subset \K_C$ we have $\hT_{\a,V} (t) G\in\L_{\a C}$ and
\[
\lluu   \hT_{\a,V} (t) G, k\rruu  =\lluu   G, \hT^\ast_{\a,V} (t) k\rruu ,
\]
where, by construction, $\hT^\ast_{\a,V} (t) k\in\K_{\a C}$. But
$G\in\L_C$, $k\in\K_C$ implies
\[
\lluu   \hT_{\a,V} (t) G, k\rruu  =\lluu   \hT_V (t) G, k\rruu  =\lluu G,
\hT^\ast(t) k\rruu .
\]
Hence, $\hT^\ast_V(t) k = \hT^\ast_{\a,V} (t) k\in\K_{\a C}$ that is what we need.

Since $\hT^\odot_V(t)$ and $\hT^\odot_\ren(t)$  are restrictions of
$\hT^\ast_V(t)$ and  $\hT^\ast_\ren(t)$ onto
$\overline{D(\hL_V^\ast)}$ and $\overline{D(\hLrenadj)}$,
correspondingly, one has, by \eqref{inclusion}, that the
corresponding semigroups coincide on $\K_\aC$. Therefore, $\K_\aC$
is $\hT ^\odot_V(t)$- and $\hT^\odot_\ren(t)$-invariant, $\eps>0$;
and the result follows from the continuity of operators which formed
semigroups.\end{proof}

Let now  $\hT^\adot_V(t)$ and $\hT^\adot_\ren(t)$ be restrictions of
the strongly continuous semigroups $\hT^\odot_V(t)$ and
$\hT^\odot_\ren(t)$  (which acting on the Banach spaces
$\overline{D(\hL_V^\ast)}$ and $\overline{D(\hLrenadj)}$,
correspondingly) onto the closed linear subspace $\KK$ of all these
Banach spaces which are invariant w.r.t. all these
$\odot$-semigroups. By the general result (see, e.g.,
\cite[Subsection~II.2.3]{EN2000}), $\hT ^\adot_V(t)$ and
$\hT^\adot_\ren(t)$ are strongly continuous semigroups on $\KK$ with
generators $\hL^\adot_V$ and $\hL ^\adot_\ren$ which are
restrictions of the corresponding operators $\hL_V^\odot$ and
$\hL_\ren^\odot$. Namely,
\begin{align}
D(\hL^\adot_V)&=\bigl\{k\in
\KK \bigm| \hL_V^\ast k\in\KK  \bigr\},\label{domVltimes}\\
D(\hL^\adot_\ren)&=\bigl\{k\in \KK \bigm| \hLrenadj k\in \KK
\bigr\},\quad \eps>0,\label{domrentimes}
\end{align}
and
\begin{align}
\hL^\adot_V k &= \hL^\ast_V k, \qquad  k\in D(\hL^\adot_V),\\
\hL^\adot_\ren k &= \hLrenadj k, \qquad k\in D(\hL
^\adot_\ren).\label{restrren}
\end{align}

By Proposition~\ref{descsemigroupexist}, $\hT_V(t)$ and $\hT
_\ren(t)$ are contraction semigroups on $\L_C$, then, $\hT '_V(t)$
and $\hT '_\ren(t)$ are also contraction semigroups on $(\L_C)'$;
but isomorphism \eqref{isometry} is isometrical, therefore, $\hT
_V^\ast(t)$ and $\hT_\ren^\ast(t)$ are contraction semigroups on
$\K_C$. As a result, their restrictions $\hT^\adot_V(t)$ and $\hT
^\adot_\ren(t)$ are contraction semigroups on $\KK$.

To summarize, we have the Banach space $\KK$ and the family of the
strongly continuous contraction semigroups $\hT^\adot_V(t)$ and $\hT
^\adot_\ren(t)$, $\eps>0$ on this space. The generators of these
semigroups are satisfied \eqref{domVltimes}--\eqref{restrren}.
Moreover, by construction, $\hT^\adot_V(t) k = \hT^\ast_V(t) k$ and
$\hT^\adot_\ren (t) k = \hT^\ast_\ren(t) k$ for any $k\in\KK$.

\begin{theorem}\label{maintheorem}
Let $C, z, \beta, \alpha_1$ be as in Proposition~\ref{sun-inv}.
Suppose additionally that $\bar{\phi}:=\sup_\X \phi \left( x\right)
<+\infty $. Then, for any $\a\in(\a_1;1)$ and for any $k\in\K_\aC$
\begin{equation}
\bigl\Vert \hT^\adot_\ren(t)k - \hT
^\adot_V(t)k\bigr\Vert_{\K_C}\leq \eps t A \Vert k \Vert_{\K_\aC},
\quad \eps>0,
\end{equation}
where $A$ is depend on $\a$, $C$, $\bar{\phi}$ only.
\end{theorem}
\begin{proof}
Let $\hQ_\delta^\ast$, $\hP_{\delta,\eps}^\ast$, $\delta\in(0;1)$,
$\eps>0$ be the images of the dual operators $\hQ_\delta'$, $\hP
_{\delta,\eps}'$ under the isometrical isomorphism \eqref{isometry}.
Since the norms of dual operators are equal we have that $\hQ
_\delta^\ast$ and $\hP_{\delta,\eps}^\ast$ are linear contractions
on $\K_C$. Moreover, for any $k\in\K_\aC$ we have
\begin{align*}
&\int_{\Gamma_{0}}(\hQ_{\delta }G)\left( \eta \right) k\left( \eta
\right) d\lambda \left( \eta \right)  \\
=&\int_{\Gamma_{0}}\sum_{\xi \subset \eta }\left( 1-\delta \right)
^{\left\vert \xi \right\vert }\int_{\Gamma_{0}}\left( z\delta
\right) ^{\left\vert \omega \right\vert }G\left( \xi \cup \omega
\right) \\&\qquad \times e_{\lambda }\left( -E^{\phi }\left( \cdot
,\omega \right) ,\eta \setminus \xi \right) d\lambda \left( \omega
\right) k\left( \eta \right) d\lambda \left( \eta
\right)  \\
=&\int_{\Gamma_{0}}\int_{\Gamma_{0}}\int_{\Gamma_{0}}\left( 1-\delta
\right)^{\left\vert \xi \right\vert }\left( z\delta \right)
^{\left\vert \omega \right\vert }G\left( \xi \cup \omega \right) \\
&\qquad \times e_{\lambda }\left( -E^{\phi }\left( \cdot ,\omega
\right) ,\eta \right) d\lambda \left( \omega \right) k\left( \eta
\cup \xi \right) d\lambda \left( \eta \right) d\lambda
\left( \xi \right)  \\
=&\int_{\Gamma_{0}}\int_{\Gamma_{0}}\sum_{\omega \subset \xi }\left(
1-\delta \right)^{\left\vert \xi \setminus \omega \right\vert
}\left( z\delta \right)^{\left\vert \omega \right\vert }G\left( \xi
\right)
\\&\qquad \times e_{\lambda }\left( -E^{\phi }\left( \cdot ,\omega \right) ,\eta \right)
k\left( \eta \cup \xi \setminus \omega \right) d\lambda \left( \eta
\right) d\lambda \left( \xi \right)
\end{align*}%
and, therefore,
\begin{align}
(\hQ_{\delta }^{\ast }k)\left( \eta \right) =&\sum_{\omega \subset
\eta }\left( 1-\delta \right)^{\left\vert \eta \setminus \omega
\right\vert }\left( z\delta \right)^{\left\vert \omega \right\vert
}\nonumber\\ &\qquad \times \int_{\Gamma_{0}}e_{\lambda }\left(
-E^{\phi }\left( \cdot ,\omega \right) ,\xi \right) k\left( \xi \cup
\eta \setminus \omega \right) d\lambda \left( \xi \right)
.\label{dualaprrVl}
\end{align}
Then, by \eqref{verysmallz-2},
\begin{align*}
&\left( \alpha C\right)^{-\left\vert \eta \right\vert }\left\vert (\hQ %
_{\delta }^{\ast }k)\left( \eta \right) \right\vert  \\
\leq &\left\Vert k\right\Vert_{\mathcal{K}_{\alpha C}}\left( \alpha
C\right)^{-\left\vert \eta \right\vert }\sum_{\omega \subset \eta
}\left( 1-\delta \right)^{\left\vert \eta \setminus \omega
\right\vert }\left( z\delta \right)^{\left\vert \omega \right\vert
}\\&\qquad \times \int_{\Gamma_{0}}e_{\lambda }\left( E^{\phi
}\left( \cdot ,\omega \right) ,\xi \right) \left( \alpha C\right)
^{\left\vert \xi \cup \eta \setminus \omega
\right\vert }d\lambda \left( \xi \right)  \\
=&\left\Vert k\right\Vert_{\mathcal{K}_{\alpha C}}\sum_{\omega
\subset \eta }\left( 1-\delta \right)^{\left\vert \eta \setminus
\omega \right\vert }\left( \frac{z\delta }{\alpha C}\right)
^{\left\vert \omega \right\vert }\exp \left\{ \alpha
C\int_{\mathbb{R}^{d}}E^{\phi }\left( x,\omega \right)
dx\right\}  \\
=&\left\Vert k\right\Vert_{\mathcal{K}_{\alpha C}}\sum_{\omega
\subset \eta }\left( 1-\delta \right)^{\left\vert \eta \setminus
\omega \right\vert }\left( \frac{z\delta }{\alpha C}\right)
^{\left\vert \omega \right\vert
}\exp \left\{ \alpha C\left\vert \omega \right\vert \beta \right\}  \\
\leq &\left\Vert k\right\Vert_{\mathcal{K}_{\alpha C}}\sum_{\omega
\subset \eta }\left( 1-\delta \right)^{\left\vert \eta \setminus
\omega \right\vert
}\delta^{\left\vert \omega \right\vert }=\left\Vert k\right\Vert_{\mathcal{%
K}_{\aC}}.
\end{align*}
Therefore, $\K_\aC$ is $\hQ_V^\ast$-invariant, hence, $\KK$ is also
$\hQ _V^\ast$-invariant due to continuity of $\hQ_V^\ast$; moreover,
$\hQ_V^\ast$ is a contraction in $\L_\aC$. Absolutely in the same
way we may obtain that for any $k\in\K_\aC$
\begin{align}
(\hP_{\delta ,\varepsilon }^{\ast }k)\left( \eta \right)
=&\sum_{\omega \subset \eta }\left( 1-\delta \right)^{\left\vert
\eta \setminus \omega \right\vert }\left( z\delta \right)
^{\left\vert \omega \right\vert
 }e_\la\left(
e^{-\varepsilon E^{\phi }\left( \cdot ,\omega \right) },\eta
\setminus \omega \right) \nonumber\\&\quad
\times\int_{\Gamma_{0}}e_{\lambda }\left( \frac{e^{-\varepsilon
E^{\phi }\left( \cdot ,\omega \right) }-1}{\varepsilon },\xi \right)
k\left( \xi \cup \eta \setminus \omega \right) d\lambda \left( \xi
\right)\label{dualaprrren}
\end{align}
and that the set $\K_\aC$, and, therefore, the set $\KK$ are $\hP
_{\delta ,\varepsilon }^{\ast }$-invariant; moreover, $\hP
_{\delta ,\varepsilon }^{\ast }$ is a contraction in $\L_\aC$. We preserve the same
notations for the restrictions of this contractions onto $\KK$.

Now, for any fixed $\eps>0$ we consider a set $D_\eps:=\bigl\{k\in
\K_\aC \bigm| \hLrenadj k\in \KK  \bigr\}$. By \eqref{domrentimes},
$D_\eps$ is a core for the operator $\hL ^\adot_\ren$. Next, let us
show that for any $k\in D_\eps$
\begin{equation}\label{apprrenest}
\lim_{\delta\goto 0}\Bigl\Vert \frac{1}{\delta}( \hP
^\ast_{\delta,\eps}-\1 )k - \hL^\adot_\ren k\Bigr\Vert_{\K_C} =0.
\end{equation}
Indeed, let
\begin{align*}
(\hP_{\delta ,\varepsilon }^{\ast, (0) }k)\left( \eta \right) =&(1-\delta)^\n k(\eta);\\
(\hP_{\delta ,\varepsilon }^{\ast, (1) }k)\left( \eta \right)
=&\sum_{x\in \eta }\left( 1-\delta \right)^{\left\vert \eta
\right\vert -1} z\delta e_\la\left( e^{-\varepsilon E^{\phi }\left(
\cdot ,x \right) },\eta \setminus x \right) \nonumber\\&\qquad
\times\int_{\Gamma_{0}}e_{\lambda }\left( \frac{e^{-\varepsilon
E^{\phi }\left( \cdot ,x \right) }-1}{\varepsilon },\xi \right)
k\left( \xi \cup \eta \setminus x \right) d\lambda \left( \xi
\right);
\end{align*}
and $\hP_{\delta ,\varepsilon }^{\ast, (\geq 2) } = \hP_{\delta
,\varepsilon }^{\ast } - \hP_{\delta ,\varepsilon }^{\ast, (0) } -
\hP_{\delta ,\varepsilon }^{\ast, (1) }$. One may improve inequality
\eqref{usefulineq}, namely, for any $n\in\N\cup\{0\}$,
$\delta\in(0;1)$
\begin{align*}
0 \leq n-   \frac{1-(1-\delta)^n}{\delta}\leq\delta \frac{n(n-1)}{2}.
\end{align*}
Then, for any $k\in\K_\aC$, $\eta\neq\emptyset$
\begin{align}
&C^{-\n}\biggl\vert\frac{1}{\delta}( \hP^{\ast,(0)}_{\delta,\eps}-\1 )k(\eta) + |\eta|k(\eta)\biggr\vert\\
\leq& \Vert k\Vert_{\K_\aC} \a^\n \Bigl\vert \n-
\frac{1-(1-\delta)^\n}{\delta}\Bigr\vert\leq \frac{\delta}{2} \Vert
k\Vert_{\K_\aC} \a^\n \n (\n-1)\nonumber
\end{align}
and the function $\a^x x(x-1)$ is bounded for $x\geq 1$, $\a\in(0;1)$.
Next, for any $k\in\K_\aC$, $\eta\neq\emptyset$
\begin{align}
&C^{-\n}\Biggl\vert\frac{1}{\delta} \hP
^{\ast,(1)}_{\delta,\eps}k(\eta) -z\sum_{x\in\eta}\int_{\Ga_0}
e_{\la  }\left( e^{-\eps \phi \left( x-\cdot \right) },\eta\setminus
x \right) \label{genrenadj1}\\&\qquad\times e_{\la }\left(
\frac{e^{-\eps \phi \left( x-\cdot \right) }-1}{\eps },\xi   \right)
k\left( \xi \cup \eta\setminus x\right) d\la(\xi)\Biggr\vert\nonumber\\
\leq & \frac{z}{\aC}\a^\n\sum_{x\in\eta}\bigl(\left( 1-\delta
\right)^{\left\vert \eta \right\vert -1} -1\bigr)e_\la \left(
e^{-\varepsilon E^{\phi }\left( \cdot ,x \right) },\eta \setminus x
\right) \nonumber\\&\qquad\times\int_{\Gamma_{0}}e_{\lambda }\left(
\a C \frac{|e^{-\varepsilon
E^{\phi }\left( \cdot ,x \right) }-1|}{\varepsilon },\xi \right)  d\lambda \left( \xi \right)\nonumber\\
 \leq &\frac{z}{\aC} \a^\n\sum_{x\in\eta}\bigl|\left( 1-\delta \right)^{\left\vert \eta
\right\vert -1} -1\bigr| \exp{\{\a C \beta\}}\nonumber\\
\leq& \frac{z}{\aC} \a^\n \delta \n (\n-1)  \exp{\{\a C
\beta\}}\nonumber
\end{align}
that is smaller then $\delta$ uniformly in $\n$. And, finally,
\begin{align}
&\frac{1}{\delta C^{\left\vert \eta \right\vert }} \sum_{\substack{
\omega \subset \eta  \crcr \left\vert \omega \right\vert \geq 2}}
\left( 1-\delta \right)^{\left\vert \eta \setminus \omega
\right\vert }\left( z\delta \right)^{\left\vert \omega \right\vert
}e_\la\left( e^{-\varepsilon E^{\phi }\left( \cdot ,\omega \right)
},\eta \setminus \omega \right)\nonumber\\&\qquad
\times\int_{\Gamma_{0}} e_{\lambda }\left(
\Biggl\vert\frac{e^{-\varepsilon E^{\phi }\left( \cdot ,\omega
\right) }-1}{\varepsilon }\Biggr\vert,\xi \right) |k( \xi
\cup \eta \setminus \omega ) | d\lambda \left( \xi \right) \nonumber \\
\leq &\frac{1}{\delta C^{\left\vert \eta \right\vert }}\sum
_{\substack{ \omega \subset \eta  \crcr \left\vert \omega \right\vert \geq 2}}%
\left( 1-\delta \right)^{\left\vert \eta \setminus \omega
\right\vert }\left( z\delta \right)^{\left\vert \omega \right\vert
}\nonumber\\&\qquad \times \int_{\Gamma_{0}}e_{\lambda }\left(
E^{\phi }\left( \cdot ,\omega \right) ,\xi \right) \left( \alpha
C\right)^{\left\vert \xi \right\vert }\left( \alpha
C\right)^{\left\vert \eta \right\vert -\left\vert \omega \right\vert
}d\lambda
\left( \xi \right) \nonumber\\
=&\alpha^{\left\vert \eta \right\vert }\frac{1}{\delta }
\sum_{\substack{ \omega \subset \eta  \crcr \left\vert \omega
\right\vert \geq 2}}\left( 1-\delta \right)^{\left\vert \eta
\setminus \omega \right\vert }\left( \frac{z\delta }{\alpha C}\exp
\left\{ \alpha C\beta \right\} \right) ^{\left\vert \omega
\right\vert } \nonumber
\end{align}
but recall that $\a>\a_1$, therefore, $z\exp\{\aC\beta\}\leq \aC$,
and one may continue
\begin{align}
\leq &\alpha ^{\left\vert \eta
\right\vert }\frac{1}{\delta }\sum
_{\substack{ \omega \subset \eta  \crcr \left\vert \omega \right\vert \geq 2}}%
\left( 1-\delta \right)^{\left\vert \eta \setminus \omega
\right\vert }\delta^{\left\vert \omega \right\vert }=\delta \alpha
^{\left\vert \eta \right\vert }\sum_{k=2}^{\left\vert \eta
\right\vert }\frac{\left\vert \eta \right\vert !}{k!\left(
\left\vert \eta \right\vert -k\right) !}\left(
1-\delta \right)^{\left\vert \eta \right\vert -k}\delta^{k-2} \nonumber\\
=&\delta \alpha^{\left\vert \eta \right\vert }\sum_{k=0}^{\left\vert
\eta \right\vert -2}\frac{\left\vert \eta \right\vert !}{\left(
k+2\right) !\left( \left\vert \eta \right\vert -k-2\right) !}\left(
1-\delta \right)
^{\left\vert \eta \right\vert -k-2}\delta^{k} \nonumber\\
=&\delta \alpha^{\left\vert \eta \right\vert }\left\vert \eta
\right\vert \left( \left\vert \eta \right\vert -1\right)
\sum_{k=0}^{\left\vert \eta \right\vert -2}\frac{\left( \left\vert
\eta \right\vert -2\right) !}{\left( k+2\right) !\left( \left\vert
\eta \right\vert -k-2\right) !}\left( 1-\delta
\right)^{\left\vert \eta \right\vert -2-k}\delta^{k} \nonumber\\
\leq &\delta \alpha^{\left\vert \eta \right\vert }\left\vert \eta
\right\vert \left( \left\vert \eta \right\vert -1\right)
\sum_{k=0}^{\left\vert \eta \right\vert -2}\frac{\left( \left\vert
\eta
\right\vert -2\right) !}{k!\left( \left\vert \eta \right\vert -k-2\right) !}%
\left( 1-\delta \right)^{\left\vert \eta \right\vert -2-k}\delta
^{k}\nonumber\\=& \delta \alpha^{\left\vert \eta \right\vert
}\left\vert \eta \right\vert \left( \left\vert \eta \right\vert
-1\right).\nonumber
\end{align}
Combining these inequalities, we obtain \eqref{apprrenest}.

Analogously, one may obtain that for any $k\in D_V:=\bigl\{k\in
\K_\aC \bigm| \hL_V^\ast k\in \KK  \bigr\}$ (that is core for $\hL
^\adot_V$)
\begin{equation}\label{apprrenest1}
\lim_{\delta\goto 0}\Bigl\Vert \frac{1}{\delta}( \hQ
^\ast_{\delta}-\1 )k - \hL^\adot_V k\Bigr\Vert_{\K_C} =0.
\end{equation}

By Lemma~\ref{EK_res-conv}, we obtain that for any $k\in\KK$
\[
(\hP ^\ast_{\delta,\eps})^{\bigl[\frac{t}{\delta}\bigr]}k\goto\hT
^\adot_\ren(t)k; \qquad (\hQ
^\ast_{\delta})^{\bigl[\frac{t}{\delta}\bigr]}k\goto\hT ^\adot_V(t)k
\]
(convergence in $\KK$, recall that norm in this space is $\|\cdot\|_{\K_C}$).

Therefore, to use the same arguments as in the proof of
Theorem~\ref{descsemigroupconv} and to apply
Lemma~\ref{powersofcontractions}, we need only to show that for any
$k\in\K_\aC$
\begin{equation}\label{dopex}
\bigl\Vert \hP^\ast_{\delta,\eps}k- \hQ^\ast_{\delta}k
\bigr\Vert_{\K_C}\leq \eps\delta A\|k\|_{\K_\aC}.
\end{equation}

We have the following elementary inequalities. For any $\left\{
a_{k}\right\} _{k=1}^{n}\subset \lbrack 0;1]$, $n\in \N$
\begin{equation}\label{specbern}
1-\prod\limits_{k=1}^{n}a_{k}\leq \sum_{k=1}^{n}\left(
1-a_{k}\right),
\end{equation}
which can be easily checked by the induction principle. Next, since
\[
x+e^{-x}-1\leq x^{2},\quad x\geq 0,
\]%
we obtain%
\begin{equation}\label{longBern}
E^{\phi }\left( x,\omega \right) \left( 1-\frac{1-e^{-\eps E^{\phi
}\left( x,\omega \right) }}{\eps E^{\phi }\left( x,\omega \right) }%
\right) \leq\eps \left( E^{\phi }\left( x,\omega \right) \right)
^{2}.
\end{equation}%

Hence,
\begin{align*}
& \frac{1}{C^{\left\vert \eta \right\vert }}\sum_{\omega \subset
\eta }\left( 1-\delta \right)^{\left\vert \eta \setminus \omega
\right\vert
}\left( z\delta \right)^{\left\vert \omega \right\vert } \\
 & \qquad \times \int_{\Gamma_{0}}\left\vert \left(
e^{-\varepsilon E^{\phi }\left( \cdot ,\omega \right) },\eta
\setminus \omega \right) e_{\lambda }\left( \frac{e^{-\varepsilon
E^{\phi }\left( \cdot ,\omega \right) }-1}{\varepsilon },\xi \right)
-e_{\lambda }\left(
-E^{\phi }\left( \cdot ,\omega \right) ,\xi \right) \right\vert \\
 &\qquad \qquad \times k\left( \xi \cup \eta \setminus \omega
\right)
d\lambda \left( \xi \right)  \\
  \leq &\frac{\left\Vert k\right\Vert_{\mathcal{K}_{\alpha C}}}{%
C^{\left\vert \eta \right\vert }}\sum_{\omega \subset \eta }\left(
1-\delta \right)^{\left\vert \eta \setminus \omega \right\vert
}\left( z\delta
\right)^{\left\vert \omega \right\vert } \\
 & \qquad \times \int_{\Gamma_{0}}e_{\lambda }\left( E^{\phi
}\left( \cdot ,\omega \right) ,\xi \right) \left\vert \left(
e^{-\varepsilon E^{\phi }\left( \cdot
,\omega \right) },\eta \setminus \omega \right) e_{\lambda }\left( \frac{%
1-e^{-\varepsilon E^{\phi }\left( \cdot ,\omega \right)
}}{\varepsilon E^{\phi }\left( \cdot ,\omega \right) },\xi \right)
-1\right\vert \\
&  \qquad \times\left( \alpha C\right)^{\left\vert \xi \cup \eta
\setminus \omega \right\vert }d\lambda \left( \xi \right)\\
\intertext{and, by \eqref{specbern}, one may continue }
  \leq & \alpha^{\left\vert \eta \right\vert }\left\Vert k\right\Vert_{%
\mathcal{K}_{\alpha C}}\sum_{\omega \subset \eta }\left( 1-\delta
\right) ^{\left\vert \eta \setminus \omega \right\vert }\left(
z\delta \right) ^{\left\vert \omega \right\vert }\int_{\Gamma
_{0}}e_{\lambda }\left( E^{\phi }\left( \cdot ,\omega \right) ,\xi
\right) \\ &\qquad \times\sum_{x\in \eta \setminus \omega }\left(
1-e^{-\varepsilon E^{\phi }\left( x,\omega \right) }\right) \left(
\alpha C\right)^{\left\vert \xi \setminus \omega \right\vert
}d\lambda \left( \xi \right)  \\ &  +\,\alpha^{\left\vert \eta \right\vert }\left\Vert k\right\Vert_{\mathcal{K%
}_{\alpha C}}\sum_{\omega \subset \eta }\left( 1-\delta
\right)^{\left\vert \eta \setminus \omega \right\vert }\left(
z\delta \right)^{\left\vert \omega \right\vert }\int_{\Gamma
_{0}}e_{\lambda }\left( E^{\phi }\left(
\cdot ,\omega \right) ,\xi \right) \\ &\qquad \times \sum_{x\in \xi }\left( 1-\frac{%
1-e^{-\varepsilon E^{\phi }\left( x,\omega \right) }}{\varepsilon
E^{\phi }\left( x,\omega \right) }\right) \left( \alpha
C\right)^{\left\vert \xi \setminus \omega \right\vert }d\lambda
\left( \xi \right)\\
\intertext{and, by \eqref{longBern},}
 \leq & \alpha^{\left\vert \eta \right\vert }\left\Vert k\right\Vert_{%
\mathcal{K}_{\alpha C}}\sum_{\omega \subset \eta }\left( 1-\delta
\right)
^{\left\vert \eta \setminus \omega \right\vert }\left( \frac{z\delta }{%
\alpha C}\exp \left\{ \alpha C\beta \right\} \right)^{\left\vert
\omega \right\vert }\sum_{x\in \eta \setminus \omega }\varepsilon
E^{\phi }\left( x,\omega \right)  \\
 & +\,\alpha^{\left\vert \eta \right\vert }\left\Vert k\right\Vert_{\mathcal{K%
}_{\alpha C}}\sum_{\omega \subset \eta }\left( 1-\delta
\right)^{\left\vert \eta \setminus \omega \right\vert }\left(
\frac{z\delta }{\alpha C}\right)
^{\left\vert \omega \right\vert }\\ & \qquad \times \int_{\Gamma_{0}}\int_{\mathbb{R}%
^{d}}\varepsilon \left( E^{\phi }\left( x,\omega \right) \right)
^{2}e_{\lambda }\left( E^{\phi }\left( \cdot ,\omega \right) ,\xi
\right) \left( \alpha C\right)^{\left\vert \xi \right\vert }\alpha
Cdxd\lambda \left( \xi \right)\\
\intertext{again, $z\exp\{\aC\beta\}\leq\aC$ and we continue}
 \leq  & \eps \, \bar{\phi}\,\alpha^{\left\vert \eta \right\vert }\left\Vert k\right\Vert_{%
\mathcal{K}_{\alpha C}}\sum_{\omega \subset \eta }\left( 1-\delta
\right) ^{\left\vert \eta \setminus \omega \right\vert
}\delta^{\left\vert \omega
\right\vert }|\eta \setminus \omega | \cdot |\omega| \\
 & + \,  \eps\alpha C \, \bar{\phi}\alpha^{\left\vert \eta \right\vert }\left\Vert k\right\Vert_{\mathcal{K%
}_{\alpha C}}\sum_{\omega \subset \eta }\left( 1-\delta
\right)^{\left\vert \eta \setminus \omega \right\vert }\delta
^{\left\vert \omega \right\vert } \left\vert \omega
\right\vert^2=:J.
\end{align*}
To complete the proof we need to use the following simple estimates:
for any $\lv  \xi \rv  =n\geq 2$ one has
\begin{align}
&\sum_{\omega \subset \xi }\lv  \omega \rv  \lv  \xi \setminus
\omega \rv  \left( 1-\delta \right)^{\lv \xi
\setminus \omega \rv  }\delta^{\lv  \omega \rv  } \label{est1}\\
=&\sum_{k=1}^{n-1}\frac{n!}{k!\left( n-k\right) !}k\left( n-k\right)
\left(
1-\delta \right)^{n-k}\delta^{k} \nonumber\\
=&\sum_{k=1}^{n-1}\frac{n!}{\left( k-1\right) !\left( n-k-1\right)
!}\left(
1-\delta \right)^{n-k}\delta^{k} \nonumber\\
=&\sum_{k=0}^{n-2}\frac{n!}{k!\left( n-\left( k+1\right) -1\right)
!}\left(
1-\delta \right)^{n-\left( k+1\right) }\delta^{k+1} \nonumber\\
=&\delta \left( 1-\delta \right) n\left( n-2\right) \sum_{k=0}^{n-2}\frac{%
\left( n-2\right) !}{k!\left( n-2-k\right) !}\left( 1-\delta \right)
^{n-2-k}\delta^{k} \nonumber\\
=&\delta \left( 1-\delta \right) n\left( n-2\right) \left( 1-\delta
+\delta \right)^{n-2}\leq \delta \cdot 2^{n}=\delta \cdot 2^{\lv \xi
\rv  }\nonumber
\end{align}%
(and this estimate is trivial for $\lv  \xi \rv  \leq 1$); and,
for any $n=\lv  \xi \rv  \geq 1$%
\begin{align}
&\sum_{\omega \subset \xi }\left( 1-\delta \right)^{\lv  \xi
\setminus \omega \rv  }\delta^{\lv  \omega \rv
}\lv  \omega \rv ^{2} \label{est2}\\
=&\sum_{k=1}^{n}\frac{n!}{k!\left( n-k\right) !}k^{2}\left( 1-\delta
\right)^{n-k}\delta^{k} \nonumber\\
=&\delta \sum_{k=1}^{n}\frac{n!}{\left( k-1\right) !\left(
n-1-\left( k-1\right) \right) !}k\left( 1-\delta \right)^{\left(
n-1\right) -\left(
k-1\right) }\delta^{k-1} \nonumber\\
=&\delta \sum_{k=0}^{n-1}\frac{n!}{k!\left( n-1-k\right) !}k\left(
1-\delta
\right)^{\left( n-1\right) -k}\delta^{k} \nonumber\\
\leq &\delta n\left( n-1\right) \left( 1-\delta +\delta \right)
^{n-1}<\delta \cdot 2^{n}\nonumber
\end{align}%
(and, again, it is trivial  for $\xi =\emptyset $).

Then, by \eqref{est1}, \eqref{est2}, we obtain for any $|\eta|\geq
2$
\[
J\leq \eps \, \bar{\phi}\,\alpha^{\left\vert \eta \right\vert
}\left\Vert k\right\Vert_{
\mathcal{K}_{\alpha C}}\delta \n(\n-1)+  \eps\alpha C \, \bar{\phi}\alpha^{\left\vert \eta \right\vert }\left\Vert k\right\Vert_{\mathcal{K%
}_{\alpha C}}\delta \n(\n-1)\leq\eps \delta A,
\]
where $A$ is independent on $\eta$. \end{proof}

\begin{corollary}\label{maincor}
Let the conditions of Theorem~\ref{maintheorem} hold. Then for any
$\{k^\e, k\}\subset\K_\aC$, $\eps>0$
\begin{equation}\label{estforeps}
\bigl\Vert \hT^\adot_\ren(t)k^\e - \hT
^\adot_V(t)k\bigr\Vert_{\K_C}\leq \bigl\Vert k^\e -
k\bigr\Vert_{\K_C} + \eps t A \Vert k \Vert_{\K_\aC}.
\end{equation}
\end{corollary}
\begin{proof}
The proof follows directly from the triangle inequality and the
contractive property of the semigroup $\hT^\adot_\ren$.
\end{proof}

And now we will show that our Vlasov limiting dynamics has the
properties described in the Subsection~\ref{scalingdescr}.

\begin{theorem}\label{Vlasovscheme}
Let $C, z, \beta, \alpha_1$ be as in Proposition~\ref{sun-inv}, and
$\a_2:=\max\bigl\{\a_1,\frac{z}{C}\bigr\}\in(0;1)$. Let $\rho_0$ be
a measurable function on $\X$ such that there exists $\a\in(\a_2;1)$
such that $0\leq\rho_0(x)\leq \aC$ for a.a. $x\in\X$. Then the
Cauchy problem
\begin{equation}\label{CauchyVlasov}
\begin{cases}
\dfrac{\partial}{\partial t} k_t = \hL^\ast_V k_t\\
k_0=e_\la (\rho_0)
\end{cases}
\end{equation}
is well-defined on $\KK$ and has a solution
$k_t=e_\la(\rho_t)\in\K_\aC$, where $\rho_t$ is a  solution of the
Cauchy problem
\begin{equation}\label{CauchyVlasoveqn}
\begin{cases}
\dfrac{\partial}{\partial t} \rho_t(x) = -\rho_t(x) + z \exp\biggl\{\displaystyle-\int_\X \rho_t(y) \phi(x-y)dy\biggr\}, \\
\rho_t \bigr|_{t=0}(x)=\rho_0(x),
\end{cases}
\end{equation}
for a.a. $x\in\X$ such that $0\leq\rho_t(x)\leq \aC$ for a.a. $x\in\X$.
\end{theorem}
\begin{proof}
First of all, we note that \eqref{verysmallz-2} implies $z<C$,
therefore, the condition $\frac{z}{C}<1$ holds. Next, if
\eqref{CauchyVlasoveqn} has a solution $\rho_t(x)\geq0$ then
$\frac{\partial}{\partial t} \rho_t(x) \leq -\rho_t(x) + z $ and,
therefore, $\rho_t(x)\leq r_t(x)$ where $r_t(x)$ is a solution of
the Cauchy problem
\begin{equation*}\label{CauchyEst}
\begin{cases}
\dfrac{\partial}{\partial t} r_t(x) = -r_t(x) + z , \\
r_t \bigr|_{t=0}(x)=\rho_0(x),
\end{cases}
\end{equation*}
for a.a. $x\in\X$, hence,
\begin{equation*}
r_t(x)=e^{-t}\rho_0(x)+z(1-e^{-t})=z+e^{-t}(\rho_0(x)-z)
\leq \max\{z,\rho_0(x)\}\leq \aC,
\end{equation*}
that yields $0\leq\rho_t(x)\leq \aC$.

To prove the existence of the solution of
\eqref{CauchyVlasoveqn}
let us fix some $T>0$
and define the Banach space $X_T=C([0;T],L^\infty(\X))$
of all continuous functions on $[0;T]$ with
values in $L^\infty(\X)$; the norm on $X_T$ is given by
\mbox{$\|u\|_T:=\max\limits_{t\in[0;T]}\|u_t\|_{L^\infty(\X)}$}.
We denote by $X_T^+$ the cone of the all
nonnegative
functions
from $X_T$.

Let $\Phi$ be a  mapping which
assign to any $v\in X_T$ the solution
$u_t$
of the linear Cauchy problem
\begin{equation}\label{CauchyLin}
\begin{cases}
\dfrac{\partial}{\partial t} u_t(x) = -u_t(x) + z \exp\{\displaystyle-(v_t*\phi)(x)\}, \\
u_t \bigr|_{t=0}(x)=\rho_0(x),
\end{cases}
\end{equation}
for a.a. $x\in\X$, where we use the usual
notation for convolution on $\X$:\linebreak $(f*g)(x):=\int_\X
f(y)g(x-y)dy$. Therefore,
\begin{equation}\label{defPhi}
(\Phi v)_t(x)=e^{-t}\rho_0(x)+z\int_0^te^{-(t-s)}\exp\{-(v_t*\phi)(x)\}ds\geq0.
\end{equation}
Similarly as before we obtain that $v\in X_T^+$ implies the estimate
$|(\Phi v)_t(x)|\leq\max\{z,\rho_0(x)\}$; in~particular, $\Phi v\in
X_T^+$. Next, using elementary inequality $|e^{-a}-e^{-b}|\leq|a-b|$
for any $a,b\geq0$, we obtain that for any $v, w\in X_T^+$
\begin{align*}
\bigl| (\Phi v)_t(x)-(\Phi w)_t(x) \bigr|&\leq z\int_0^te^{-(t-s)}\Bigl|\exp\{-(v_t*\phi)(x)-\exp\{-(w_t*\phi)(x)\}\Bigr|ds\\
&\leq z\int_0^te^{-(t-s)}\bigl|(v_t*\phi)(x)-(w_t*\phi)(x)\bigr|ds\\
&\leq z\int_0^te^{-(t-s)}(|v_t-w_t|*\phi)(x) ds\\
&\leq z\beta\|v-w\|_T(1-e^{-t}),
\end{align*}
where we used the inequality $|(f*g)(x)|\leq
\|f\|_{L^\infty(\X)}\|g\|_{L^1(\X)}$ and condition
\eqref{integrability}. Therefore, $\|\Phi v-\Phi w\|_T\leq
z\beta\|v-w\|_T$. Since \eqref{verysmallz-2} implies $z\beta\leq
e^{-1}$ (see~the proof of Proposition~\ref{sun-inv}), hence, $\Phi$
is a contraction mapping on the cone $X_T^+$. Taking, as usual,
$v^{(n)}=\Phi^nv^{(0)}$, $n\geq1$ for $v^{(0)}\in X_T^+$ we obtain
that $\{v^{(n)}\}\subset X_T^+$ is a fundamental sequence in $X_T$
which has, therefore, a unique limit point $v\in X_T$. Since $X_T^+$
is a closed cone we have that $v\in X_T^+$. Then, identically to the
classical Banach fixed point theorem, $v$ will be a fixed point of
$\Phi$ on $X_T$ and a unique fixed point on $X_T^+$. Then, this $v$
is the nonnegative solution of \eqref{CauchyVlasoveqn} on the
interval $[0;T]$. By the note above, $v_t(x)\leq \aC$. Changing
initial value in \eqref{CauchyVlasoveqn} onto $\rho_t
\bigr|_{t=T}(x)=v_T(x)$ we may extend all our considerations on the
time-interval $[T;2T]$ with the same estimate $v_t(x)\leq \aC$; and
so on. As a a result, \eqref{CauchyVlasoveqn} has a global bounded
solution $\rho_t(x)$ on $\R_+$.

Clearly, $k_0=e_\la(\rho_0)\in\K_\aC\subset\KK$. Then
$k_t=\hT_V^\adot(t)k_0$ will be a strongly differentiable function
(in the sense of norm $\|\cdot\|_{\K_C}$ in $\KK$); moreover,
$k_t\in\K_\aC$. Next, if we substitute $k_t=e_\la(\rho_t)$ into
\eqref{CauchyVlasov}, then, by \eqref{genrenadj}, we obtain
\begin{align*}
&\sum_{x\in\eta} \frac{\partial}{\partial t} \rho_t(x)
e_\la(\rho_t,\eta\setminus x)\\=&-\lv\eta\rv e_\la(\rho_t,\eta)
\\& +z\sum_{x\in\eta}e_\la(\rho_t,\eta\setminus x)\int_{\Ga_0} e_{\la
}\left(- \phi \left( x-\cdot \right),\xi \right) e_\la(\rho_t,\xi)
d\la(\xi)\\=&-\sum_{x\in\eta} \rho_t(x) e_\la(\rho_t,\eta\setminus
x)\\&+z\sum_{x\in\eta}e_\la(\rho_t,\eta\setminus
x)\exp\biggl\{-\int_\X \phi(x-y)\rho_t(y)dy\biggr\},
\end{align*}
that holds since $\rho_t$ is satisfied \eqref{CauchyVlasoveqn}.
\end{proof}

\begin{remark}\label{RemarkKirkwoodMonroe}
Note that the stationary equation for \eqref{CauchyVlasoveqn} has
the following form
\begin{equation}\label{KirkwoodMonroe}
\rho(x) = z \exp\biggl\{-\int_\X \rho(y) \phi(x-y)dy\biggr\}
\end{equation}
and coincides with the famous Kirkwood--Monroe equation
(\cite{KM1941}, see also, e.g., \cite{GK1976} and references
therein, and the recent work \cite{CP2010}).
\end{remark}

\subsection{Further considerations}
We have realized the scheme proposed at the end of
Subsection~\ref{scalingdescr}. But let us explain also the rigorous
meaning of the equivalence \eqref{ordersing} which was background to
all our consideration.

Let $C, z, \beta, \alpha_2$ be as in Theorem~\ref{Vlasovscheme}.
Then, for any fixed $\eps>0$ we have $1-\exp\{-\eps\phi\}\in
L^1(\X)$ and, by \cite[Proposition~3.2]{FKKZ2010}, $\hL_\eps$, given
by \eqref{Lhateps}, is a linear operator in $\L_{\eps^{-1}C}$ with
dense domain $\L_{2\eps^{-1}C}$. Consider the image $\bigl(
\hL^\ast_\eps, D(\hL_\eps^\ast)\bigr)$ in
$\K_{\eps^{-1}C}=R_{\eps^{-1}}\K_C$ under the isometrical
isomorphism $R_{\eps^{-1}C}$ of the dual operator $\bigl(\hL'_\eps,
D(\hL'_\eps)\bigr)$ in $(\L_{\eps^{-1}C})'$.

We are not able to show that $\hL_\eps$ is a generator of a strongly
continuous semigroup in $\L_{\eps^{-1}C}$ since a condition like
\eqref{verysmallparam} (with $\eps^{-1}C$ instead of $C$) cannot be
fulfilled uniformly in $\eps>0$. But one can do in the following
manner.

Let $\a\in(\a_2;1)$  and let us consider the space
$\mathbb{K}^\a_\eps=\overline{\K_{\eps^{-1}\aC}}^{\K_{\eps^{-1}C}}$.
Note that for any $r^\e\in\mathbb{K}^\a_\eps$ there exist
$\{r^\e_n\}\subset\K_{\eps^{-1}\aC}$ such that
\[
0=\lim_{n\goto\infty}\|r^\e_n-r^\e\|_{\K_{\eps^{-1}C}}
=\lim_{n\goto\infty}\|R_\eps r^\e_n-R_\eps r^\e\|_{\K_{C}}
\]
and the inclusion $R_\eps r^\e_n \in\K_\aC$, $n\in\N$ yields $R_\eps
r^\e\in \KK$. Vise versa, for any $k^\e\in\KK$ we see that
$R_{\eps^{-1}}k^\e\in\mathbb{K}^\a_\eps$. As a result, $R_\eps$
provides an isometrical isomorphism between the Banach spaces
$\mathbb{K}^\a_\eps$ and $\KK$. Then,
$U_\eps^\a(t):=R_{\eps^{-1}}\hT_\ren^\adot (t)R_\eps $
 will be a strongly continuous contraction semigroup on $\mathbb{K}^\a_\eps$
 with the generator $A_\eps^\a =R_{\eps^{-1}}\hL_\ren^\adot R_\eps$ and the domain
  $D(A_\eps^\a)=R_{\eps^{-1}}D(\hL_\ren^\adot)$.
 Moreover, since $\K_\aC \cap D(\hL_\ren^\adot)$ is a core for
 $\hL_\ren^\adot$, the set $K_{\eps^{-1}\aC}\cap D(A_\eps^\a)$ is a core for $A_\eps^\a$ and on this core the operator $A_\eps^\a$ coincides
with\ $\hL_\eps^\ast$. Note that, the semigroup $U_\eps^\a(t)$ is
the rigorous analog of $\hT_\eps^\ast$ in \eqref{ordersing}.

Let now $\{k_0,k_0^\e|\eps>0\}\subset K_\aC$. Then, by
\eqref{estforeps},
\begin{align}
&\bigl\Vert U_\eps^\a(t) R_{\eps^{-1}} k_{0}^\e
-R_{\eps^{-1}}\hT_V^\adot(t)
k_0\bigr\Vert_{\K_{\eps^{-1}C}}\label{dop1}\\=&\bigl\Vert R_\eps
\bigl(U_\eps^\a(t) R_{\eps^{-1}} k_0^\e -R_{\eps^{-1}}\hT_V^\adot(t)
k_0\bigr)\bigr\Vert_{\K_{C}}\nonumber\\ = &\bigl\Vert\hT_\ren^\adot
(t)k_0^\e -\hT_V^\adot (t)k_0\bigr\Vert_{\K_{C}} \leq A \eps
t\|k_0\|_{\K_\aC}+\|k_0^\e-k_0\|_{\K_C}.\nonumber
\end{align}
On the other hand,
\begin{align}
&\bigl\Vert U_\eps^a(t) R_{\eps^{-1}} k_0^\e
-R_{\eps^{-1}}\hT_V^\adot
(t)k_0\bigr\Vert_{\K_{\eps^{-1}C}}\label{dop2}\\=&\esssup_{\eta\in\Ga_0}\left\{
(\eps^{-1}C)^{-\n} \bigl\vert R_{\eps^{-1}}\hT_V^\adot
(t)k_0(\eta)\bigr\vert \left\vert\frac{U_\eps^a(t) R_{\eps^{-1}}
k_0^\e (\eta)}{R_{\eps^{-1}}\hT_V^\adot (t)k_0(\eta)}-1\right\vert
\right\}\nonumber\\= &\esssup_{\eta\in\Ga_0}\left\{ C^{-\n}
\bigl\vert\hT_V^\adot (t)k_0(\eta)\bigr\vert
\left\vert\frac{U_\eps^a(t) R_{\eps^{-1}} k_0^\e
(\eta)}{R_{\eps^{-1}}\hT_V^\adot (t)k_0(\eta)}-1\right\vert
\right\}.\nonumber
\end{align}

In particular, if
\begin{equation}\label{k-conv}
\lim_{\eps\goto0}\| k_0^\e-k_0\|_{\K_C}=0
\end{equation}
then \eqref{dop1}, \eqref{dop2} imply
\begin{equation}\label{equivaa}
\lim_{\eps\goto0}\frac{U_\eps^a(t) R_{\eps^{-1}} k_0^\e
(\eta)}{R_{\eps^{-1}}\hT_V^\adot (t)k_0(\eta)}=1 \quad \mathrm{for}
\ \la\mathrm{-a.a.} \ \eta\in\Ga_0.
\end{equation}
The equality \eqref{equivaa} is a rigorous realization of the
equivalence \eqref{ordersing} (with changes $k_0^{(\eps)}$ onto
$R_{\eps^{-1}} k_0^\e $).

Moreover, let $T>0$ and suppose that there exists a function
$c:\Ga_0\rightarrow(0;+\infty)$ such that
\begin{equation}\label{condveryspec}
q(\a,T):=\sup_{t\in [0;T]}\esssup_{\eta\in\Ga_0}\frac{c(\eta)}{\hT_V^\adot (t)k_0(\eta)}<+\infty.
\end{equation}
Then, using the equality
\begin{align*} &c(\eta){C^{-\n}}
\left\vert\frac{U_\eps^a(t) R_{\eps^{-1}} k_0^\e
(\eta)}{R_{\eps^{-1}}\hT_V^\adot(t) k(\eta)}-1\right\vert \\ = &
C^{-\n} \bigl\vert\hT_V^\adot(t) k_0(\eta)\bigr\vert
\left\vert\frac{U_\eps^a(t) R_{\eps^{-1}} k_0^\e
(\eta)}{R_{\eps^{-1}}\hT_V^\adot(t) k_0(\eta)}-1\right\vert
\frac{c(\eta)}{\bigl\vert\hT_V^\adot (t)k_0(\eta)\bigr\vert},
\end{align*}
we obtain that for such $k_0$ and for any $t\in[0;T]$
\begin{equation}
\left\Vert\frac{U_\eps^a(t) R_{\eps^{-1}} k_0^\e (\eta)}{R_{\eps^{-1}}\hT_V^\adot
(t)k_0(\eta)}-1\right\Vert_{C,c} \leq q(\a,T) A \eps  t\|k_0\|_{\K_\aC}+\|k_0^\e-k_0\|_{\K_C},
\end{equation}
where
\[
\|k\|_{C,c}=\esssup_{\eta\in\Ga_0} \frac{|k(\eta)|}{C^\n
c^{-1}(\eta)}.
\]
This gives that the equivalence \eqref{ordersing} may be shown in a
proper Banach space which is independent on $\eps$.

\begin{remark}
The condition \eqref{condveryspec} on $k_0$ is reasonable: for
example, for $k_0=e_\la(\rho_0)$, since, by the
Theorem~\ref{Vlasovscheme}, we have $\hT_V^\adot
(t)k_0(\eta)=e_\la(\rho_t,\eta)$, where $\rho_t$ satisfies
\eqref{CauchyVlasoveqn}; therefore, \eqref{condveryspec} holds for
any $|\rho_0(x)|\leq \aC$ such that
\[
\sup_{t\in[0;T]}\inf_{x\in\X}|\rho_t(x)|\geq
\rho_{\min} >0
\]
if we set $c(\eta)=e_\la(\rho_{\min},\eta)=\rho_{\min}^\n$.
Moreover, we obtain that $|\rho_t(x)|\leq \aC$.
The following example shows which function $k_0^\e$ one can choose in this case.
\end{remark}

\begin{example}
Let $k_0(\eta)=\rho_0^\n$, $\rho_0 \in(0;\aC)$. Let us consider the
scaled Lebesgue--Poisson exponent
$k_0^\e(\eta)=e_\la\bigl(\rho_0(1+\eps u(\cdot)),\eta\bigr)$, where
$\sup_{x\in\X}|u(x)|=\bar{u}<\infty$, $\eps>0$. Then for any
$\eps<\frac{\aC-\rho_0}{\rho_0 \bar{u}}$ we have
$|k_0^\e(\eta)|<(\aC)^\n$. Moreover,
\begin{align*}
& C^{-\n}\bigl\vert k_0^\e(\eta)-k_0(\eta) \bigr\vert=
\Bigl(\frac{\rho_0}{C}\Bigr)^\n\bigl\vert e_\la\bigl(1+\eps
u(\cdot),\eta\bigr)-1\bigr\vert\\\leq&
\Bigl(\frac{\rho_0}{C}\Bigr)^\n \eps \sup_{s\in(0;\eps)}\biggl\vert
\frac{d}{ds} e_\la\bigl(1+s u(\cdot),\eta\bigr)\biggr\vert\\ =&
\Bigl(\frac{\rho_0}{C}\Bigr)^\n \eps \sup_{s\in(0;\eps)}\biggl\vert
\sum_{x\in\eta}u(x) e_\la\bigl(1+s u(\cdot),\eta\setminus
x\bigr)\biggr\vert\\ \leq& \Bigl(\frac{\rho_0}{C}\Bigr)^\n \eps
\sum_{x\in\eta} \bar{u} e_\la\bigl(1+\eps \bar{u},\eta\setminus
x\bigr)\\ \leq& \Bigl(\frac{\rho_0}{C}\Bigr)^\n \eps \n \bar{u}
\biggl(1+\frac{\aC-\rho_0}{\rho_0 \bar{u}} \bar{u}\biggr)^{\n-1}\\
=&\eps \frac{\rho_0}{\aC}\n\a^\n\leq \eps \frac{\rho_0}{\aC}
\frac{-1}{e\ln\a}.
\end{align*}
As a result, $\|k_0^\e-k_0\|_{\K_C}\goto0$ as $\eps\goto0$.
\end{example}

\section*{Acknowledgement}
The financial support of DFG through the SFB 701 (Bielefeld
University), German-Ukrainian Projects 436 UKR 113/94, 436 UKR
113/97 and FCT through POCI and PTDC/MAT/67965/2006 is gratefully
acknowledged.

\end{document}